\documentclass[letterpaper, 11pt]{article}
\usepackage{appendix}
 
\newcommand{\comment}[1]{}
 
\usepackage{graphicx}

\usepackage{subfigure}
\usepackage{amsmath}
\usepackage{amssymb}
\usepackage{mdwlist}

\usepackage{algorithm}

\usepackage{hyperref}

\usepackage{xspace}
\usepackage{setspace}

\newcommand{\sv}{{\mathcal V}}
\newcommand{\se}{{\mathcal E}}

\newcommand{\sn}{{\mathcal N}}
 
\newenvironment{proof}{\noindent {\bf Proof:}~}{\hspace*{\fill}\(\Box\)}
\newenvironment{proofSketch}{\noindent{\bf Proof Sketch:}}{\hspace*{\fill}\(\Box\)}
\newtheorem{theorem}{Theorem}

\newtheorem{claim}{Claim}
\newtheorem{corollary}{Corollary}
\newtheorem{definition}{Definition}

\newtheorem{lemma}{Lemma}

\newcommand{\fillbox}{\hspace*{\fill}\(\Box\)}

\def\noflash#1{\setbox0=\hbox{#1}\hbox to 1\wd0{\hfill}}

\newcommand{\scripte}{\mathcal{E}}
\newcommand{\scriptv}{\mathcal{V}}

\newcommand{\Zightarrow}{\rightarrow}

\newcommand{\propagate}[3]{#1\stackrel{#3}{\rightsquigarrow}{#2}}
\newcommand{\notpropagate}[3]{#1\stackrel{#3}{\not\rightsquigarrow}{#2}}

\newcommand{\Propagate}{{\tt Propagate}}
\newcommand{\Equality}{{\tt Equality}}
\newcommand{\ssss}{{\mbox{\it start}}}
\newcommand{\eeee}{{\mbox{\it end}}}

\begin{document}
\title{Byzantine Consensus in Directed Graphs\footnote{\normalsize This research is supported in part by Army Research Office grant W-911-NF-0710287. Any opinions, findings, and conclusions or recommendations expressed here are those of the authors and do not necessarily reflect the views of the funding agencies or the U.S. government.}}

\author{Lewis Tseng$^{1,3}$, and Nitin Vaidya$^{2,3}$\\~\\
 \normalsize $^1$ Department of Computer Science,\\
 \normalsize $^2$ Department of Electrical and Computer Engineering, 
 and\\ \normalsize $^3$ Coordinated Science Laboratory\\ \normalsize University of Illinois at Urbana-Champaign\\ \normalsize Email: \{ltseng3, nhv\}@illinois.edu \\ \normalsize Phone: +1 217-244-6024, +1 217-265-5414}

\date{August 24, 2012\footnote{Revised on February 18, 2014 to make major improvements to the presentation and related work.}}
\maketitle


\begin{abstract}
\normalsize

Consider a synchronous point-to-point network of $n$ nodes connected by {\em directed} links, wherein each node has a binary input. 
This paper proves a {\em tight} necessary and sufficient condition on the underlying communication topology for
achieving Byzantine consensus among these nodes in the presence of
up to $f$ Byzantine faults. We derive a necessary condition, and then we provide a constructive proof of
sufficiency by presenting a Byzantine consensus algorithm for directed graphs
that satisfy the necessary condition. 

~

Prior work has developed analogous necessary and sufficient conditions
for {\em undirected} graphs. It is known that, for undirected graphs,
the following two conditions are
together necessary and sufficient \cite{impossible_proof_lynch, welch_book,dolev_82_BG}:
(i) $n\geq 3f+1$, and (ii) network connectivity greater than $2f$.
However, these conditions are not adequate to completely characterize Byzantine consensus
in {\em directed} graphs.

\end{abstract}

~

~


~


\thispagestyle{empty}
\newpage
\setcounter{page}{1}

\section{Introduction}
\label{s_intro}

In this work, we explore algorithms for achieving Byzantine consensus \cite{psl_BG_1982} in a synchronous point-to-point network in the presence of Byzantine faulty nodes. The network is modeled as a {\em directed} graph, i.e., the communication links between neighboring nodes are not necessarily bi-directional. Our work is 
motivated by the presence of directed links in wireless networks.
However, we believe that the results here are of independent interest as well.

The Byzantine consensus problem \cite{psl_BG_1982} considers $n$ nodes, of which at most $f$ nodes may be faulty. The faulty nodes may deviate from the algorithm in arbitrary fashion. Each node has an {\em input} in $\{0,1\}$. A Byzantine consensus algorithm is {\em correct} if it satisfies the following three properties:

\begin{itemize}
\item \textbf{Agreement}: the output (i.e., decision) at all the fault-free nodes is identical.

\item \textbf{Validity}: the output of every fault-free node equals the input of a fault-free node.

\item \textbf{Termination}: every fault-free node eventually decides on an output.

\end{itemize}

In networks with undirected links (i.e., in undirected graphs), it is well-known that the following two conditions together are both necessary and sufficient for the existence of Byzantine consensus algorithms \cite{impossible_proof_lynch, welch_book, dolev_82_BG}: (i) $n \geq 3f+1$, and (ii) node connectivity
greater than $2f$. The first condition, that is, $n\geq 3f+1$, is necessary for directed graphs as well. Under the second condition, each pair of nodes in the undirected graph can
communicate {\em reliably} with each other. In particular, either a given pair of nodes is connected
directly by an edge, or there are $2f+1$ node-disjoint paths between the pair of nodes. However, reliable communication between every pair of node is {\em not} necessary for achieving consensus in directed graphs. In Section \ref{s_discussion}, we address this statement in more details.

This paper presents {\em tight} necessary and sufficient conditions for Byzantine consensus in {\em directed} graphs. We provide a constructive proof of sufficiency by presenting a Byzantine consensus algorithm for directed graphs satisfying the necessary condition. The rest of the paper is organized as follows. Section \ref{s_related} discusses the related work. Section \ref{s_term} introduces our system model and some terminology used frequently in our presentation. The main result and the implications are presented in Section \ref{nec_2}. The Byzantine consensus algorithm for directed graphs is described, and its correctness is also proved in Section \ref{s_sufficiency}. The paper summarizes in Section \ref{s_conclusion}.

\section{Related Work}
\label{s_related}

Lamport, Shostak, and Pease addressed the Byzantine agreement problem in \cite{psl_BG_1982}. Subsequent work \cite{impossible_proof_lynch, dolev_82_BG} characterized the necessary and sufficient conditions under which the problem is solvable in {\em undirected} graphs. However, as noted above, these conditions are not adequate to fully characterize the {\em directed} graphs in which Byzantine consensus is feasible. In this work, we identify {\em tight} necessary and sufficient conditions for Byzantine consensus in {\em directed} graphs. The necessity proof presented in this paper is based on the state-machine approach, which was originally developed for conditions in undirected graphs \cite{impossible_proof_lynch,dolev_82_BG,welch_book}; however, due to the nature of directed links, our necessity proof is a non-trivial extension. The technique is also similar to the {\em withholding} mechanism, which was developed by Schmid, Weiss, and Keidar \cite{impossible_link} to prove impossibility results and lower bounds for the number of nodes for synchronous consensus under {\em transient link} failures in {\em fully-connected} graphs; however, we do not assume the transient fault model as in \cite{impossible_link}, and thus, our argument is more straightforward.

In related work, Bansal et al. \cite{Bansal_disc11} identified tight conditions for achieving Byzantine consensus in {\em undirected} graphs using {\em authentication}. Bansal et al. discovered that all-pair reliable communication is not necessary to achieve consensus when using authentication. Our work differs from Bansal et al. in that our results apply in the absence of authentication or any other security primitives; also our results apply to {\em directed} graphs. We show that even in the absence of authentication all-pair reliable communication is not necessary for Byzantine consensus.

Several papers have also addressed communication between a single source-receiver pair.
Dolev et al. \cite{Dolev90perfectlysecure} studied the problem of secure communication, which achieves both fault-tolerance and perfect secrecy between a single source-receiver pair in undirected graphs, in the presence of node and link failures. Desmedt and Wang considered the same problem in directed graphs \cite{yvo_eurocrypt02}. In our work, we do not consider secrecy, and address the consensus problem rather than the single source-receiver pair problem. Shankar et al. \cite{Shankar_SODA08} investigated reliable communication between a source-receiver pair in directed graphs allowing for an arbitrarily small error probability in the presence of a Byzantine failures. Our work addresses deterministically correct algorithms for consensus.

Our recent work \cite{vaidya_PODC12,Tseng_general,vaidya_incomplete} has considered a restricted class of iterative algorithms for achieving {\em approximate} Byzantine consensus in directed graphs, where fault-free nodes must agree on values that are approximately equal to each other using iterative algorithms with limited memory. The conditions developed in such prior work are {\em not} necessary when no such restrictions are imposed. Independently, LeBlanc et al. \cite{leblanc_HiCoNs,Sundaram_journal}, and Zhang and Sundaram \cite{Sundaram,Sundaram_ACC} have developed results for iterative
algorithms for {\em approximate} consensus under a {\em weaker} fault model, where a faulty node must send
identical messages to all the neighbors. In this work, we consider the problem of {\em exact} consensus (i.e., the outputs at fault-free nodes must be exactly identical), and we do not impose any restriction on the algorithms or faulty nodes.

Alchieri et al. \cite{BFT-CUP_OPODIS} explored the problem of achieving exact consensus in {\em unknown} networks with Byzantine nodes, but the underlying communication graph is assumed to be fully-connected. In this work, the network is assumed to be known to all nodes, and may not be fully-connected.

\section{System Model and Terminology}
\label{s_term}

\subsection{System Model}
The system is assumed to be {\em synchronous}.
The synchronous communication network consisting
of $n$ nodes is modeled as a simple {\em directed} graph $G(\scriptv,\scripte)$, where $\scriptv$ is the set of $n$ nodes, and $\scripte$ is the set of directed edges between the nodes in $\scriptv$.  We assume that $n\geq 2$, since the consensus problem for $n=1$ is trivial.  Node $i$ can transmit messages to another node $j$ if and only if the directed edge $(i,j)$ is in $\scripte$. Each node can send messages to itself as well; however, for convenience, we {exclude self-loops} from set $\scripte$. That is, $(i,i)\not\in\scripte$ for $i\in\scriptv$. With a slight abuse of terminology, we will use the terms {\em edge} and {\em link}, and similarly the terms
{\em node} and {\em vertex},
 interchangeably.


All the communication links are reliable,
FIFO (first-in first-out) and deliver each transmitted message exactly once.
When node $i$ wants to send message M on link $(i,j)$ to node $j$, it
puts the message M in a send buffer for link $(i,j)$. No further operations
are needed at node $i$; the mechanisms for implementing reliable,
FIFO and exactly-once semantics 
are transparent to the nodes. When a message is delivered on link ($i,j$),
it becomes available to node $j$ in a receive buffer for link $(i,j)$.
As stated earlier, the communication network is synchronous, and thus,
each message sent on link ($i,j$) is delivered to node $j$ within a 
bounded interval of time.

\paragraph{Failure Model:}
We consider the Byzantine failure model, with up to $f$ nodes becoming faulty. A faulty node may {\em misbehave} arbitrarily. Possible misbehavior includes sending incorrect and mismatching (or inconsistent) messages to different neighbors. The faulty nodes may potentially collaborate with each other. Moreover, the faulty nodes are assumed to have a complete knowledge of the execution of the algorithm, including the states of all the nodes, contents of messages the other nodes send to each other, the algorithm specification, and the network topology.

\subsection{Terminology}

We now describe terminology that is used frequently in our presentation. Upper case italic letters are used below to name subsets of $\sv$,
and lower case italic letters are used to name nodes in $\sv$.

\paragraph{Incoming neighbors:}

\begin{itemize}
\item Node $i$ is said to be an incoming neighbor of node $j$ if $(i,j)\in \se$.
\item
For set $B\subseteq \sv$, node $i$ is said to be an incoming
neighbor of set $B$ if $i\not\in B$, and there exists $j\in B$
such that $(i,j)\in \se$. Set $B$ is said to have $k$ incoming neighbors in set $A$ if set $A$ contains $k$ distinct incoming neighbors of $B$.

\end{itemize}

\paragraph{Directed paths:}
All paths used in our discussion are directed paths.

\begin{itemize}
\item Paths from a node $i$ to another node $j$:
\begin{itemize}
\item For a directed path from node $i$ to node $j$, node $i$ is said to be the ``source node'' for the path.
\item An ``$(i,j)$-path'' is a directed path from node $i$ to node $j$. An ``$(i,j)$-path excluding $X$'' is a directed path from node $i$ to node $j$ that does not contain any node from set $X$. 

\item Two paths from node $i$ to node $j$ are said to be ``disjoint'' 
	if the two paths only have nodes $i$ and $j$
	in common, with all remaining
	nodes being distinct.
	
\item The phrase ``$d$ disjoint $(i,j)$-paths'' refers to
	$d$ pairwise disjoint paths from node $i$ to node $j$. The phrase ``$d$ disjoint $(i,j)$-paths excluding $X$'' refers to $d$ pairwise disjoint $(i,j)$-paths that do not contain any node from set $X$.

\end{itemize}
\item Every node $i$ trivially has a path to itself. That is,
	for all $i\in\sv$, an $(i,i)$-path
	excluding $\sv-\{i\}$ exists.
\item Paths from a set $S$ to node $j\not\in S$:

\begin{itemize}
\item A path is said to be an ``$(S,j)$-path'' if it is an
	$(i,j)$-path for some $i\in S$. An ``$(S,j)$-path excluding $X$'' is a $(S,j)$-path that does not contain any node from set $X$.

\item Two $(S,j)$-paths are said to be ``disjoint'' 
	if the two paths only have node $j$
	in common, with all remaining
	nodes being distinct (including the source nodes on the paths).
	
\item The phrase ``$d$ disjoint $(S,j)$-paths'' refers to
	$d$ pairwise disjoint $(S,j)$-paths. The phrase ``$d$ disjoint $(S,j)$-paths excluding $X$'' refers to $d$ pairwise disjoint $(S,j)$-paths that do not contain any node from set $X$.

\end{itemize}
\end{itemize}

\paragraph{Graph Properties:}

\begin{definition}
\label{def:propagate}
Given disjoint subsets $A,B,F$ of $\sv$ such that $|F|\leq f$,
set $A$ is said to
\underline{propagate in $\sv-F$} to set $B$
if either (i) $B=\emptyset$, or (ii) for each node $b\in B$, there exist at least
$f+1$ disjoint $(A,b)$-paths excluding $F$.
\end{definition}
We will denote the fact that set $A$ propagates in $\sv-F$ to set $B$
by the notation 
\[
\propagate{A}{B}{\sv-F}
\]
When it is not true that $\propagate{A}{B}{\sv-F}$, we will denote that fact by 
\[
\notpropagate{A}{B}{\sv-F}
\]

~

For example, consider Figure \ref{f:2-core} below when $f = 2$ and $F=\{u_1, u_2\}, A = \{u_3, u_4, u_5, u_6, u_7\}$ and $B = \{w_1, w_2, w_3, w_4, w_5, w_6, w_7\}$, then $\propagate{B}{A}{\sv-F}$ and $\notpropagate{A}{B}{\sv-F}$.

\begin{figure}[hbtp!]
\centering
\includegraphics[scale=0.65, bb=-100 -30 749 240]{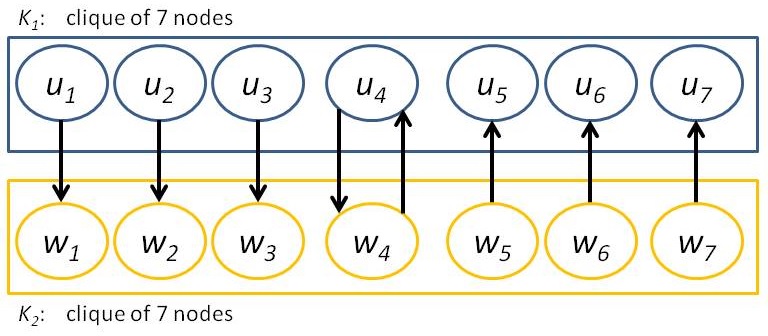}
\vspace*{-30pt}
\caption{A network tolerating $2$ faults. Edges inside cliques $K_1$ and $K_2$ are not shown.}
\label{f:2-core}

\end{figure}

\begin{definition}
\label{def:G-F}
For $F\subset \sv$, graph $G_{-F}$ is obtained by removing from $G(\sv,\se)$
all the nodes in $F$, and all the links incident on nodes in $F$.
\end{definition}

\begin{definition}
\label{def:strong}
A subgraph $S$ of $G$ is said to be {\em strongly connected}, if for all nodes $i, j$ in $S$, there exists an $(i,j)$-path in $G$. 
\end{definition}

\section{Main Result}
\label{nec_2}



We now present the main result of this paper.

\begin{theorem}
\label{t_nec_2}
Byzantine consensus is possible
in $G(\scriptv,\scripte)$
{\bf if and only if} for any node partition $A, B, F$ of $\scriptv$, where $A$ and $B$ are both non-empty, and $0\leq |F| \leq f$, either $\propagate{A}{B}{\sv-F}$ or $\propagate{B}{A}{\sv-F}$.
\end{theorem}

\begin{proof}
Appendix \ref{a_1and2} presents the proof of necessity of the condition in the theorem. In Appendix \ref{a_1and2}, we first prove the necessity of an alternate form of the condition using the state-machine approach developed in prior work \cite{impossible_proof_lynch,dolev_82_BG,welch_book}. We then prove that the alternate necessary condition is equivalent to the condition stated in Theorem \ref{t_nec_2}.

In Section \ref{s_sufficiency}, we present a constructive proof of sufficiency of the condition in the theorem. In particular, we present a Byzantine consensus algorithm and prove its correctness in all directed graphs that satisfy the condition stated in Theorem \ref{t_nec_2}. 
\end{proof}

\subsection{Implications of the necessary and sufficient condition}
\label{s_discussion}

Here, we discuss some interesting implications of Theorem \ref{t_nec_2}.

\begin{itemize}
\item Lower bounds on number of nodes and incoming neighbors are identical to the case in undirected networks \cite{impossible_proof_lynch,dolev_82_BG,welch_book}:

This observation is not surprising, since undirected graphs are a special case of directed graphs.

\begin{corollary}
\label{cor:2f+1}
Suppose that a correct Byzantine consensus algorithm exists for
$G(\scriptv,\scripte)$.
Then,
 (i) $n \geq 3f+1$, and
 (ii) if $f>0$, then each node must have at least $2f+1$ incoming neighbors.
\end{corollary}
\begin{proof}
The proof is in Appendix \ref{a_cor:2f+1}.
\end{proof}


\item Reliable communication between {\bf all} node pairs is {\bf not} necessarily required:

This observation is also not surprising, but nevertheless
interesting (because, in undirected graphs, Byzantine consensus
is feasible if and only if all node pairs can communicate with each
other reliably). To illustrate the above observation, consider the simple
example in Figure \ref{f:core}, with $f=1$. In Figure \ref{f:core},
nodes $v_1,v_2,v_3,v_4$ have directed links to each other, forming a 4-node
clique -- the links inside the clique are not shown in the figure. Node $x$ does not have a directed link to any other node, but has links from the
other 4 nodes.
Yet, Byzantine consensus can be achieved easily by first reaching consensus within the
4-node clique, and then propagating the consensus value (for the 4-node consensus) to node $x$.  
Node $x$ can choose majority of the values received from the nodes in the
4-node clique as its own output.
It should be easy to see that this algorithm works correctly for inputs in $\{0,1\}$ as
required in the Byzantine consensus formulation considered in this work.

\vspace*{-40pt}
\begin{figure}[hbtp!]
\centering


\includegraphics[scale=0.6, bb=-250 -30 749 240]{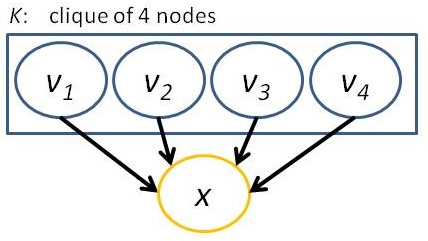}
\vspace*{-30pt}
\caption{A network tolerating $1$ fault. Edges inside clique $K$ are not shown.}
\label{f:core}
\end{figure}

\item For a cut $(A,B)$ of the communication graph, there may not necessarily be $2f+1$ disjoint links 
in any one direction (i.e., from nodes in $A$ to nodes in $B$, or vice-versa):

The above observation is surprising, since it suggests that reliable communication may not be
feasible in either direction across a given cut in the communication graph.
We illustrate this using the system in
Figure \ref{f:2-core} in Section \ref{s_term}, which contains 
two cliques $K_1$ and $K_2$, each containing 7 nodes.
Within each clique, each node has a directed link to the other 6 nodes
in that clique -- these links
are not shown in the figure.
There are 8 directed links with one endpoint in clique $K_1$ and the other endpoint in clique $K_2$.
We prove in Appendix \ref{a_l_2clique} that Byzantine consensus can be
achieved in this system with $f=2$.
However, there are only 4 directed links from $K_1$ to $K_2$,
and 4 directed links from $K_2$ to $K_1$. Thus, reliable communication is
{\em not} guaranteed across the cut $(K_1,K_2)$ in either direction.
Yet, Byzantine consensus is achievable using Algorithm BC.  
In Appendix \ref{a_l_2clique}, we present a family of graphs, named 2-clique network, which satisfies the
condition in Theorem \ref{t_nec_2}. Figure \ref{f:2-core} shown in Section \ref{s_term} is the 2-clique network for $f=2$.
Section \ref{s_sufficiency} proves that Byzantine consensus is possible in all graphs
that satisfy the necessary condition. Therefore, consensus is possible in the 2-clique
network as well. 
\end{itemize}

\section{Sufficiency: Algorithm BC and Correctness Proof}
\label{s_sufficiency}

In this section, we assume that graph $G(\sv,\se)$ satisfies the condition stated in Theorem \ref{t_nec_2}, even if this is not stated explicitly again. We present Algorithm BC (Byzantine Consensus) and prove its correctness in all graphs that satisfy the condition in Theorem \ref{t_nec_2}. This proves that the necessary condition is also sufficient. 
When $f=0$, all the nodes are fault-free,
and as shown in Appendix \ref{a_f_0}, 
the proof of sufficiency is trivial.
In the rest of our discussion below, we will assume that $f>0$. 

The proposed Algorithm BC is presented below.
Each node $i$ maintains two state variables that are
explicitly used in our algorithm: $v_i$ and $t_i$.
Each node maintains other state as well (such as
the routes to other nodes); however, we do not introduce
additional notation for that for simplicity.

\begin{itemize}

\item
{Variable $v_i$:}
Initially, $v_i$ at any node $i$ is equal to the binary input at node $i$.
During the execution of the algorithm, $v_i$ at node $i$ may be updated
several times. Value $v_i$ at the end of the algorithm
represents node $i$'s decision (or output)
for Algorithm BC. The output at each node is
either 0 or 1. At any time during the execution of the algorithm,
the value $v_i$ at node $i$ is said to be {\em valid}, if it equals some fault-free node's input. Initial value $v_i$ at a fault-free node $i$ is valid, because it equals its own input.
Lemma \ref{l_valid_1} proved later in Section \ref{s:correct} implies that $v_i$ at a fault-free node $i$ always remains valid
throughout the execution of Algorithm BC.

\item
{Variable $t_i$:} Variable $t_i$ at any node $i$ 
may take a value in $\{0,1,\perp\}$, where $\perp$ is distinguished
from 0 and 1.
Algorithm BC makes use of procedures \Propagate~and \Equality~that are described soon below. These procedures take $t_i$ as input, and possibly also modify $t_i$.
Under some circumstances, as discussed later, state variable $v_i$ at node $i$ is set equal to $t_i$,
in order to update $v_i$.
\end{itemize}

Algorithm BC consists of two loops, an OUTER loop, and an INNER loop. The OUTER loop is performed for
each subset of nodes $F$, $|F|\leq f$. 
For each iteration of the OUTER loop, many iterations of the INNER loop are performed.
The nodes in $F$ do not participate in any of these INNER loop iterations.
For a chosen $F$,
each iteration of the INNER loop is performed for a different partition of $\sv-F$.

Since there are at most $f$ faults, one iteration of the OUTER loop has $F$ exactly equal
to the set of faulty nodes. Denote the actual set of faulty nodes as $F^*$. 
Algorithm BC has two properties, as proved later:
\begin{itemize}
\item State $v_i$ of each fault-free node $i$ at the end of any particular INNER loop iteration equals the state of some
fault-free node at the start of that INNER loop iteration. 
Thus, Algorithm BC ensures that the state $v_i$ of each fault-free node $i$ remains valid at all times.
\item
By the end of the OUTER loop iteration for $F=F^*$, all the fault-free nodes reach agreement.
\end{itemize}
The above two properties ensure that, when Algorithm BC terminates, the validity and agreement properties are both satisfied. 

Each iteration of the INNER loop, for a given set $F$, considers a partition $A,B$ of the nodes in $\sv-F$
such that $\propagate{A}{B}{\sv-F}$.
Having chosen a partition $A,B$, intuitively speaking, the goal of the INNER loop iteration is for the nodes in set $A$ to attempt to influence the state of the nodes in the other partition.
A suitable set $S\subseteq A \cup B$ is identified and agreed a priori using the known topology information. There are two possible cases. In Case 1 in Algorithm BC, $S \subseteq A$, and nodes in $S$ use procedure \Equality~(step (b) in the pseudo-code) to decide the value to propagate to nodes in $\sv-F-S$ (step (c)). In Case 2, $S \subseteq A \cup B$, and nodes in $S$ first learn the states at nodes in $A$ using procedure \Propagate~(step (f)), and then use procedure \Equality~(step (g)) to decide the value to propagate to nodes in $\sv-F-S$ (step (h)). These steps ensure that if $F = F^*$, and nodes in $A$ have the same $v$ value, then $S$ will propagate that value, and all nodes in $\sv - F^* - S$ (Case 1: step (d)) or in $\sv - F^* - (A \cap S)$ (Case 2: step (i)) will set $v$ value equal to the value propagated by $S$, and thus, the agreement is achieved. As proved later, in at least one INNER loop iteration with $F = F^*$, nodes in $A$ have the same $v$ value.

~

\hrule

\vspace*{2pt}

\noindent {\bf Algorithm BC}

\vspace*{4pt}

\hrule

\vspace*{4pt}

~

\indent {\em Comment}: Note that Algorithm BC can be implemented distributedly if every node has prior knowledge of the topology. For the convenience of reader, the pseudo-code below is presented in a centralized fashion. \\

\noindent
(OUTER LOOP)\\ For each $F \subset \scriptv$, where $0 \leq |F| \leq f$:

  \begin{list}{}{}
  \item (INNER LOOP)\\For each partition $A,B$ of $\sv - F$ such
	 that $A, B$ are non-empty, and $\propagate{A}{B}{\sv-F}$:

	STEP 1 of INNER loop:

	\begin{itemize}
	\item {\bf Case 1:}
            {\bf if} $\propagate{A}{B}{\sv-F}$ and $\notpropagate{B}{A}{\sv-F}$:

	Choose a non-empty set $S\subseteq A$ such that $\propagate{S}{\sv-F-S}{\sv-F}$,	and $S$ is strongly connected in $G_{-F}$ ($G_{-F}$ is defined in Definition \ref{def:G-F}).

        \begin{list}{}{}
        \item[(a)] At each node $i\in S:~~~$ $t_i:=v_i$ 
	\item[(b)] \Equality($S$)               
	\item[(c)] \Propagate($S,\sv-F-S$)
	\item[(d)] At each node $j\in \sv-F-S:~~~$ if $t_j\neq\perp$, then $v_j:=t_j$
        \end{list}

	\item {\bf Case 2:}
            {\bf if} $\propagate{A}{B}{\sv-F}$ and $\propagate{B}{A}{\sv-F}$:

	Choose a non-empty set $S\subseteq A\cup B$ such that	$\propagate{S}{\sv-F-S}{\sv-F}$,	$S$ is strongly connected in $G_{-F}$, and $\propagate{A}{(S-A)}{\sv-F}$.

        \begin{list}{}{}
        \item[(e)] At each node $i\in A:~~~$ $t_i=v_i$
        \item[(f)] \Propagate($A,S-A$)
	\item[(g)] \Equality($S$)
	\item[(h)] \Propagate($S,\sv-F-S$)
	\item[(i)] At each node $j\in \sv-F-(A\cap S):~~~$ if $t_j\neq\perp$,
			 then $v_j:=t_j$
	\end{list}
	\end{itemize}  

	STEP 2 of INNER loop:
        \begin{list}{}{}
	\item[(j)] Each node $k\in F$ receives $v_j$ from
		each $j\in N_k$, where $N_k$ is a set consisting of
		$f+1$ of $k$'s incoming neighbors in $\sv-F$.
		If all the received values are identical, then $v_k$ is
		set equal to this identical value; else $v_k$ is unchanged.
	\end{list}
   \end{list}

\hrule

\subsection{Procedure \Propagate($P,D$)}
\label{ss_propagate}

\Propagate($P,D$) assumes that $P\subseteq \sv-F$, $D\subseteq \sv-F$,
$P\cap D=\emptyset$ and $\propagate{P}{D}{\sv-F}$. Recall that set $F$ is the set chosen in each OUTER loop as specified by Algorithm BC.

\vspace*{2pt}
\hrule
\vspace*{2pt}
\noindent \Propagate($P,D$)
\vspace*{4pt}
\hrule
\vspace*{-2pt}
\begin{list}{}{}
\item[(1)]
Since $\propagate{P}{D}{\sv-F}$, for each $i\in D$, there
exist at least $f+1$ disjoint ($P,i)$-paths that exclude $F$.
The source node of each of these paths is in $P$.
On each of $f+1$ such disjoint paths, the source node for that path, say $s$, sends $t_s$ to node $i$. 
 Intermediate nodes on these paths forward
received messages as necessary. 

When a node does not receive an expected message,
the message content is assumed to be $\perp$.

\item[(2)] When any node $i\in D$ receives $f+1$ values along the $f+1$
disjoint paths above:\\
 if the $f+1$ values are all equal to 0, then
$t_i:=0$; else if the $f+1$ values are all equal to 1, then
$t_i:=1$; else $t_i:=\perp$. \hfill (Note that $:=$ denotes the assignment operator.)
\end{list}
For any node $j\not\in D$, $t_j$ is not modified during \Propagate($P,D$).
Also, for any node $k\in\sv$, $v_k$ is not modified during \Propagate($P,D$).

\vspace*{4pt}
\hrule

\subsection{Procedure \Equality($D$)}
\label{ss_equality}

\Equality($D$) assumes that $D\subseteq \sv-F$, $D\neq \emptyset$,
and for each pair of nodes $i,j\in D$,
an $(i,j)$-path excluding $F$ exists, i.e., $D$ is strongly connected
in $G_{-F}$ ($G_{-F}$ is defined in Definition \ref{def:G-F}).


\vspace*{2pt}
\hrule
\vspace*{2pt}
\noindent \Equality($D$)
\vspace*{4pt}
\hrule
\begin{list}{}{}
\item[(1)] Each node $i\in D$ sends $t_i$ to all other nodes
in $D$ along paths excluding $F$.
\item[(2)] Each node $j\in D$ thus receives messages from all nodes in $D$.
Node $j$ checks whether values received from all the nodes in $D$ and
its own $t_j$ are all equal, and also belong to $\{0,1\}$.
If these conditions are {\em not} satisfied, then $t_j:=\perp$; otherwise $t_j$ is not modified.


\end{list}
For any node $k\not\in D$, $t_k$ is not modified in \Equality($D$).
Also, for any node $k\in \sv$, $v_k$ is not modified in \Equality($D$).

\vspace*{4pt}

\hrule

\subsection{INNER Loop of Algorithm BC for $f>0$}
\label{subsec:inner}
Assume that $f>0$.
For each $F$ chosen in the OUTER loop, the INNER loop of Algorithm BC examines
each partition $A,B$ of $\sv-F$ such that $A,B$ are both non-empty.
From the condition in Theorem \ref{t_nec_2}, we know that either
$\propagate{A}{B}{\sv-F}$ or $\propagate{B}{A}{\sv-F}$.
Therefore, with renaming of the sets we can ensure that 
$\propagate{A}{B}{\sv-F}$. Then, depending on the choice of $A,B,F$, two cases may occur:
 (Case 1) $\propagate{A}{B}{\sv-F}$ and $\notpropagate{B}{A}{\sv-F}$, and (Case 2) $\propagate{A}{B}{\sv-F}$ and $\propagate{B}{A}{\sv-F}$. 

In Case 1 in the INNER loop of Algorithm BC, we need to find
	a non-empty set $S\subseteq A$ such that $\propagate{S}{\sv-F-S}{\sv-F}$,
	and $S$ is strongly connected in $G_{-F}$ ($G_{-F}$ is defined in Definition \ref{def:G-F}).
In Case 2, we need to find
	a non-empty set $S\subseteq A\cup B$ such that
	$\propagate{S}{\sv-F-S}{\sv-F}$,
	$S$ is strongly connected in $G_{-F}$, and $\propagate{A}{(S-A)}{\sv-F}$. 
The following claim ensures that Algorithm BC can be executed correctly in $G$.

\begin{claim}
\label{claim:inner_loop}
Suppose that $G(\sv, \se)$ satisfies the condition stated in Theorem \ref{t_nec_2}. Then,

\begin{itemize}
\item The required set $S$ exists in both Case 1 and 2 of each INNER loop.

\item Each node in set $F$ has enough incoming neighbors in $\sv-F$ to perform step (j) of Algorithm BC with $f>0$.
\end{itemize}

\end{claim}

\begin{proof}
The proof of the first claim is proved in Appendix \ref{a_claim:S}.

Now, we prove the second claim.
Consider nodes in set $F$. As shown in Corollary \ref{cor:2f+1} in Section \ref{nec_2},
when $f>0$, each node in $\sv$ has at least $2f+1$ incoming neighbors.
Since $|F|\leq f$,
for each $k\in F$ there must exist at least
$f+2$ incoming neighbors in $\sv-F$. 
Thus, the desired set $N_k$ exists, satisfying
the requirement in step (j) of Algorithm BC.
\end{proof}

\comment{++++++++++
Appendix \ref{a_claim:S} shows that the required set $S$ exists in both the cases. 

Now, we show the following claim:

\begin{claim}
Each node in set $F$ has enough incoming neighbors in $\sv-F$ to perform step (j) of Algorithm BC with $f>0$.
\end{claim}

\begin{proof}
Consider nodes in set $F$. As shown in Corollary \ref{cor:2f+1} in Section \ref{nec_2},
when $f>0$, each node in $\sv$ has at least $2f+1$ incoming neighbors.
Since $|F|\leq f$,
for each $k\in F$ there must exist at least
$f+2$ incoming neighbors in $\sv-F$. 
Thus, the desired set $N_k$ exists, satisfying
the requirement in step (j) of Algorithm BC.
\end{proof}

+++++++++++++++}

%

%

\subsection{Correctness of Algorithm BC for $f>0$}
\label{s:correct}

Recall that by assumption, $F^*$ is the \underline{actual set of faulty nodes} in
the network ($0\leq |F^*|\leq f$).
Thus, the set of fault-free nodes is $\sv-F^*$. When discussing a certain INNER loop iteration, we sometimes add
\underline{superscripts $\ssss$\, and $\eeee$}\, to $v_i$ for node $i$ to indicate whether we are referring to $v_i$ at the start, or at the end, of that
INNER loop iteration, respectively. We first show that INNER loop preserves validity.


\begin{lemma}
\label{l_valid_1}
For any given INNER loop iteration, for each fault-free node $j\in \sv-F^*$, there exists a fault-free node $s\in \sv-F^*$ such that $v_j^\eeee=v_s^\ssss$. 
\end{lemma}

\begin{proof}
To avoid cluttering the notation, for a set of nodes $X$, we use the phrase

\hspace{1in}a fault-free node $j\in X$

\noindent
as being equivalent to

 \hspace{1in}a fault-free node $j\in X-F^*$

\noindent
because all the fault-free nodes in any set $X$ must also be in $X-F^*$.

Define \underline{set $Z$} as the set of values of $v_i$
at all fault-free $i\in\sv$ at the start of the
INNER loop iteration under consideration, i.e., $Z = \{ v_i^\ssss ~ | ~i\in \sv-F^* ~\}$.

We first prove the claim in the lemma for the fault-free nodes in $\in\sv-F$,
and then for the fault-free nodes in $F$.
Consider the following two cases in the INNER loop iteration.
\begin{itemize}
\item {\bf Case 1:} $\propagate{A}{B}{\sv-F}$ and $\notpropagate{B}{A}{\sv-F}$:

Observe that, in Case 1, $v_i$ remains unchanged for all
fault-free $i\in S$. Thus, $v_i^\eeee=v_i^\ssss$ for
$i\in S$, and hence, the claim of the lemma is trivially
true for these nodes. We will now prove the claim for fault-free $j\in\sv-F-S$.

\begin{itemize}
\item step (a):
Consider a fault-free node $i\in S$.
At the end of step (a), $t_i$ is equal to $v_i^\ssss$. Thus,
$t_i\in Z$.

\item step (b):
In step (b), step 2 of \Equality($S$) either keeps $t_i$ unchanged at
fault-free node $i\in S$
or modifies it to be $\perp$. Thus, now $t_i\in Z\cup\{\perp\}$.

\item step (c):
Consider a fault-free node $j\in \sv-F-S$. During \Propagate($S,\sv-F-S$),
$j$ receives $f+1$ values along $f+1$ disjoint paths originating
at nodes in $S$. Therefore, at least one of the $f+1$ values
is received along a path that contains only fault-free nodes;
suppose that the value received by node $j$ along this fault-free path
is equal to $\alpha$. As observed above in step (b),
$t_i$ at all fault-free nodes $i\in S$ is in $Z\cup\{\perp\}$. Thus,
$\alpha\in Z\cup\{\perp\}$.
Therefore, at fault-free node $j\in \sv-F-S$, step 2 of \Propagate($S,\sv-F-S$) will result
in $t_j\in\{\alpha,\perp\}\subseteq Z\cup\{\perp\}$.

\item step (d):
Then it follows that, in step (d), at fault-free $j\in\sv-F-S$, if $v_j$ is updated, then 
$v_j^\eeee\in Z$. On the other hand, if $v_j$ is not updated, then $v_j^\eeee=v_j^\ssss\in Z$.
\end{itemize}

\item {\bf Case 2:} $\propagate{A}{B}{\sv-F}$ and $\propagate{B}{A}{\sv-F}$:

Observe that, in Case 2, $v_j$ remains unchanged for all
fault-free $j\in A\cap S$; thus $v_j^\eeee=v_j^\ssss$
for these nodes. Now, we prove the claim in the lemma for 
fault-free $j\in\sv-F-(A\cap S)$.

\begin{itemize}

\item step (e):
For any fault-free node $i\in A$,
at the end of step (e), $t_i\in Z$.

\item
step (f):
Consider a fault-free node $m\in S-A$. During \Propagate($A,S-A$),
$m$ receives $f+1$ values along $f+1$ disjoint paths originating
at nodes in $A$. Therefore, at least one of the $f+1$ values
is received along a path that contains only fault-free nodes;
suppose that the value received by node $m$ along this fault-free path
is equal to $\gamma\in Z$.
Therefore, at node $m\in S-A$, \Propagate($A,S-A$) will result in
$t_m$ being set to a value in $\{\gamma,\perp\}\subseteq  Z\cup\{\perp\}$. Now, for $m\in S\cap A$, $t_m$ is not modified in step (f), and therefore,
for fault-free $m\in S\cap A$, $t_m\in Z$. 
Thus, we can conclude that, at the end of step (f), for all
fault-free nodes
$m\in S$, $t_m\in Z\cup\{\perp\}$.

\item
step (g):
In step (g), at each $m\in S$, \Equality($S$) either keeps $t_m$ unchanged,
or modifies it to be $\perp$. Thus, at the end of step (g), for
all fault-free $m\in S$, $t_m$ remains in $Z\cup\{\perp\}$.

\item
step (h):
Consider a fault-free node $j\in \sv-F-S$. During \Propagate($S,\sv-F-S$),
$j$ receives $f+1$ values along $f+1$ disjoint paths originating
at nodes in $S$. Therefore, at least one of the $f+1$ values
is received along a path that contains only fault-free nodes;
suppose that the value received by node $j$ along this fault-free path
is equal to $\beta$. As observed above, after step (g),
for each fault-free node $m\in S$, $t_m\in Z\cup\{\perp\}$.
Therefore, $\beta\in Z\cup\{\perp\}$, and at node $j\in \sv-F-S$,
\Propagate($S,\sv-F-S$) will result in
$t_j$ being set to a value in $\{\beta,\perp\}\subseteq Z\cup\{\perp\}$.

\item step (i):
From the discussion of steps (g) and (h) above,
it follows that, in step (i), if $v_j$ is updated
at a fault-free $j\in \sv-F-(S\cap A)$, then $v_j^\eeee\in Z$;
on the other hand, if $v_j$ is not modified,
then $v_j^\eeee=v_j^\ssss\in Z$.
\end{itemize}
\end{itemize}
 Now, consider a fault-free node $k\in F$. Step (j) uses set
$N_k\subset\sv-F$ such that $|N_k|=f+1$. As shown above, at the start of step (j), $v_j^\eeee\in Z$ at all fault-free $j\in \sv-F$. Since $|N_k| = f+1$, at least one of the nodes in $N_k$ is fault-free. Thus, of the $f+1$ values received by node $k$, at least one value must be in $Z$. It follows that if node $k$ changes $v_k$ in step (j), then the new value will also in $Z$; on the other hand, if node $k$ does not change $v_k$, then it remains equal to $v_k^\ssss \in Z$.
\end{proof}

\begin{lemma}
\label{l_validity}
Algorithm BC satisfies the validity property for Byzantine consensus.
\end{lemma}
\begin{proof}
Recall that the state $v_i$ of a fault-free node $i$ is {\em valid} if it equals the input
at a fault-free node.
For each fault-free $i\in \sv$, initially, $v_i$ is valid. 
Lemma \ref{l_valid_1} implies that after each INNER loop iteration, $v_i$ remains valid
at each fault-free node $i$.
Thus, when Algorithm BC terminates, $v_i$ at each
fault-free node $i$ will satisfy the {\em validity} property
for Byzantine consensus, as stated in Section \ref{s_intro}.
\end{proof}

\begin{lemma}
\label{l_termination}
Algorithm BC satisfies the termination property for Byzantine consensus.
\end{lemma}
\begin{proof}
Recall that we are assuming a synchronous system, and the graph $G(\sv,\se)$
is finite. Thus, Algorithm BC performs a finite number of OUTER loop iterations,
and a finite number of INNER loop iterations for each
choice of $F$ in the OUTER loop,
the number of iterations being a function of graph $G(\sv,\se)$.
Hence, the termination property is satisfied.
\end{proof}

\begin{lemma}
\label{l_agreement}
Algorithm BC satisfies the agreement property for Byzantine consensus.
\end{lemma}

\begin{proofSketch}
The complete proof is in Appendix \ref{a_l_agreement}. Recall that $F^*$ denotes the actual set of faulty nodes in the network ($0 \leq |F^*|\leq f$). Since the OUTER loop considers all possible $F\subset \sv$ such that $|F|\leq f$, eventually, the OUTER loop will be performed with $F=F^*$.
 We will show that when OUTER loop is performed with $F=F^*$, {\em agreement} is achieved.
After agreement
is reached when $F=F^*$, Algorithm BC may perform the OUTER loop with other choices of set $F$. However, due to Lemma \ref{l_valid_1}, the {\em agreement} among fault-free nodes is still preserved.
(Also, due to Lemma \ref{l_valid_1}, before the OUTER loop with $F=F^*$ is performed, $v_i$ at each fault-free node remains valid.)

Now, consider the OUTER loop with $F=F^*$.
We will say that an INNER loop iteration with $F=F^*$ is ``deciding'' if one of the following
conditions is true: (i) in Case 1 of the INNER loop iteration, after step (b) is performed,
all the nodes in set $S$ have an identical value for variable $t$,
or (ii) in Case 2 of the INNER loop iteration, after step (g) is performed,
all the nodes in set $S$ have an identical value for variable $t$.
As elaborated in Appendix \ref{a_l_agreement},
when $F=F^*$, at least one of the INNER loop
iterations must be a {\em deciding} iteration. 
Let us partition the INNER loop iterations when $F=F^*$ into three phases:
\begin{itemize}
\item Phase 1: INNER loop iterations before the first deciding iteration with $F=F^*$.
\item Phase 2: The first deciding iteration with $F=F^*$.
\item Phase 3: Remaining INNER loop iterations with $F = F^*$.
\end{itemize}

From the pseudo-code for~\Propagate~and~\Equality, observe that when $F=F^*$, all paths
used in the INNER loop iterations {\bf exclude} $F=F^*$. That is, all these paths contain only
fault-free nodes, since $F^*$ is the actual set of faulty nodes.
In each INNER loop iteration in Phase 1, we can show that value $v_i$ for each fault-free node $i$ remains unchanged from previous INNER loop iteration. 
As elaborated in Appendix \ref{a_l_agreement},
this together with fact that the value $v_i \in \{0,1\}$ for each fault-free node $i$, ensures that a {\em deciding} INNER loop iteration is eventually performed when
$F=F^*$ (e.g., when set $A$ contains the fault-free nodes with $v$ value equal to $0$, and set $B$ contains the remaining fault-free nodes, or vice-versa).
 In Phase 2, Algorithm BC achieves agreement among fault-free nodes due to the fact that nodes in set $S$ reliably propagate an identical value to all the other nodes. Finally, in Phase 3,
 due to Lemma \ref{l_valid_1},
 agreement achieved in the previous phase is preserved.
 Therefore, at the end of the OUTER loop with $F=F^*$, agreement is achieved.
\end{proofSketch}

\begin{theorem}
\label{t_correct}
Algorithm BC satisfies the agreement, validity, and termination conditions.
\end{theorem}

\begin{proof}
The theorem follows from Lemmas \ref{l_validity}, \ref{l_termination} and \ref{l_agreement}.
\end{proof}

\subsection{Application to Multi-Valued Consensus}

Algorithm BC can be used to solve a particular version of {\em multi-valued} consensus with the following properties: 

\begin{itemize}
\item \textbf{Agreement}: the output (i.e., decision) at all the fault-free nodes must be identical.

\item \textbf{Validity}: If all fault-free nodes have the same input, then the output of every fault-free node equals its input.

\item \textbf{Termination}: every fault-free node eventually decides on an output.

\end{itemize}
Under these conditions, if all the fault-free nodes do not have the same multi-valued input,
then it is possible for
the fault-free nodes to agree on a value that is not an input at any fault-free node.
This multi-valued consensus problem for $L$-bit input values can be solved by executing
$L$ instance of Algorithm BC, one instance for each bit of the input, on graphs
that satisfy the condition stated in Theorem \ref{t_nec_2}.
The 1-bit output of each of the $L$ instances put together form the $L$-bit output
of the multi-valued consensus problem. Correctness of this procedure follows from
Theorem \ref{t_correct}.

If the above validity condition for multi-valued consensus is made stronger, to require that the
output value must be the multi-valued input of a fault-free node, then the condition
in Theorem \ref{t_nec_2} is not sufficient for inputs that can take 3 or more distinct values.


\section{Conclusion}
\label{s_conclusion}

For nodes with binary inputs, we present a {\em tight} necessary and sufficient condition for achieving Byzantine consensus in synchronous {\em directed} graphs. The condition is shown to be necessary using traditional state-machine approach \cite{impossible_proof_lynch,dolev_82_BG, welch_book}. Then, we provide a constructive proof of sufficiency by presenting a new Byzantine consensus algorithm for graphs satisfying the necessary condition. The algorithm can also be used to solve multi-valued consensus.

Two open problems are of further interest: 

\begin{itemize}
\item  Algorithm BC presented in Section \ref{s_sufficiency} has exponential round complexity. The optimal round complexity for directed graphs is presently unknown.

\item It is not known whether one can efficiently determine that a given graph satisfies the condition in Theorem \ref{t_nec_2} or not.

\end{itemize}


\appendix

\newpage

\centerline{\Large\bf Appendices}



\section{Necessity Proof of Theorem \ref{t_nec_2}}
\label{a_1and2}

This appendix presents the proof of necessity of the condition stated in Theorem \ref{t_nec_2}.
We first present an alternative form of the necessary condition, named {\em Condition 1} below. We use the
well-known state-machine approach \cite{impossible_proof_lynch,dolev_82_BG, welch_book} to show the necessity
of {\em Condition 1}. Then, we prove that the condition stated in Theorem \ref{t_nec_2} is
equivalent to {\em Condition 1}.

\subsection{Necessary Condition 1}

Necessary condition 1 is stated in Theorem \ref{t_nec_1} below.
Its proof uses the familiar proof technique based on state machine approach.
Although the proof of Theorem \ref{t_nec_1} is straightforward, we include it here
for completeness.
Readers may omit the proof of Theorem \ref{t_nec_1} in this section without lack of continuity.

We first define relations $\Zightarrow$ and $\not\Zightarrow$ that are used subsequently.
These relations are defined for disjoint sets. 
Two sets are disjoint if their intersection is empty.
For convenience of presentation, we adopt the convention that
sets $A$ and $B$ are disjoint if either one of them is empty.
More than two sets are disjoint if they are pairwise disjoint.

\begin{definition}
\label{def:absorb}
For disjoint sets of nodes $A$ and $B$, where $B$ is non-empty:
\begin{itemize}
\item $A \Zightarrow B$ iff set $A$ contains at least
 $f+1$ distinct incoming neighbors of $B$.

 That is, $|~\{ i~|~(i,j)\in \se,~i\in A,~j\in B\}~| > f$.
\item $A\not\Zightarrow B$ iff $A\Zightarrow B$ is {\em not} true.
\end{itemize}
\end{definition}

The theorem below states {\em Condition 1}, and proves its necessity. 

\begin{theorem}
\label{t_nec_1}
Suppose that a correct Byzantine consensus algorithm exists for $G(\scriptv,\scripte)$.
For any partition
\footnote{Sets $X_1,X_2,X_3,...,X_p$ are said to form a partition of set $X$ provided that (i) $\cup_{1\leq i\leq p} X_i = X$, and (ii) $X_i\cap X_j=\emptyset$ if $i\neq j$.}
 $L, C, R, F$ of $\scriptv$, such that
both $L$ and $R$ are non-empty, and $|F|\leq f$, either $L\cup C\Zightarrow R$, 
or $R\cup C\Zightarrow L$.
\end{theorem}

We first describe the intuition behind the proof, followed by a formal proof. The proof is by contradiction.

Suppose that there exists a partition $L,C,R,F$ where $L,R$ are non-empty and $|F|\leq f$ such that
$C\cup R \not\Zightarrow L$, and $L\cup C \not\Zightarrow R$. Assume that the nodes in $F$ are faulty, and the nodes in sets $L, C, R$ are fault-free. Note that fault-free nodes are not aware of the identity of the faulty nodes.

Consider the case when all the nodes in $L$ have input $m$, and all the nodes in $R \cup C$ have input $M$, where $m \neq M$. Suppose that the nodes in $F$ (if non-empty) behave to nodes in $L$ as if nodes in $R\cup C\cup F$ have input $m$, while behaving to nodes in $R$ as if nodes in $L\cup C\cup F$ have input $M$. This behavior by nodes in $F$ is possible, since the nodes in $F$ are all assumed to be faulty here. 

Consider nodes in $L$.
Let $N_L$ denote the set of incoming neighbors of $L$ in $R\cup C$.
Since $R\cup C\not\rightarrow L$, $|N_L|\leq f$.
Therefore, nodes in $L$ cannot distinguish between the following two scenarios:
 (i) all the nodes in $N_L$ (if non-empty) are faulty, rest of the nodes are fault-free,  and all the fault-free nodes have input $m$, and (ii) all the nodes in $F$ (if non-empty) are faulty, rest of the nodes are fault-free, and fault-free nodes have input either $m$ or $M$. In the first scenario, for validity, the output at nodes in $L$ must be $m$. Therefore, in the second scenario as well, the output at the nodes in $L$ must be $m$. We can similarly show that the output at the nodes in $R$ must be $M$. Thus, if the condition in Theorem \ref{t_nec_1} is not satisfied, nodes in $L$ and $R$ can be forced to decide on distinct values, violating the agreement property. Now, we present the formal proof. Note that the formal proof relies on traditional state-machine approach \cite{impossible_proof_lynch, welch_book}. We include it here for completeness.

\paragraph{Proof of Theorem \ref{t_nec_1}:}
~\\

\begin{proof}
The proof is by contradiction.
Suppose that a correct Byzantine consensus algorithm, say ALGO, exists in $G(\sv, \se)$, and there exists a partition $F,L,C,R$ of $\sv$ such that $C \cup R \not\Zightarrow L$ and $L \cup C \not\Zightarrow R$. Thus, $L$ has at most $f$ incoming neighbors in $R\cup C$,
and $R$ has at most $f$ incoming neighbors in $L\cup C$.
Let us define:
\begin{eqnarray*}
N_L & = & \mbox{set of incoming neighbors of $L$ in $R\cup C$} \\
N_R & = & \mbox{set of incoming neighbors of $R$ in $L\cup C$}
\end{eqnarray*}
Then,
\begin{eqnarray}
|N_L|& \leq &f \label{e_nl} \\
|N_R|& \leq &f \label{e_nr}
\end{eqnarray}

The behavior of each node $i\in\sv$ when
using ALGO can be modeled by a state machine
that characterizes the behavior of each node $i\in\sv$. \\

We construct a new network called $\sn$, as illustrated in Figure \ref{sm_1-1}.
In $\sn$, there are three copies of each node in $C$,
and two copies of each node in $L\cup R\cup F$.
In particular, C0 represents one copy of the nodes in $C$,
C1 represents the second copy of the nodes in $C$,
and 
C2 represents the third copy of the nodes in $C$.
Similarly, R0 and R2 represent the two copies of the nodes in $R$,
L0 and L1 represent the two copies of the nodes in $L$, and
F1 and F2 represent the two copies of the nodes in $F$.
Even though the figure shows just one vertex for C1, it represents
all the nodes in $C$ (each node in $C$ has a counterpart in the nodes
represented by C1). Same correspondence holds for other vertices
in Figure \ref{sm_1-1}.


\begin{figure}[p]
\centering
\includegraphics[width=220mm,bb=0 0 960 720]{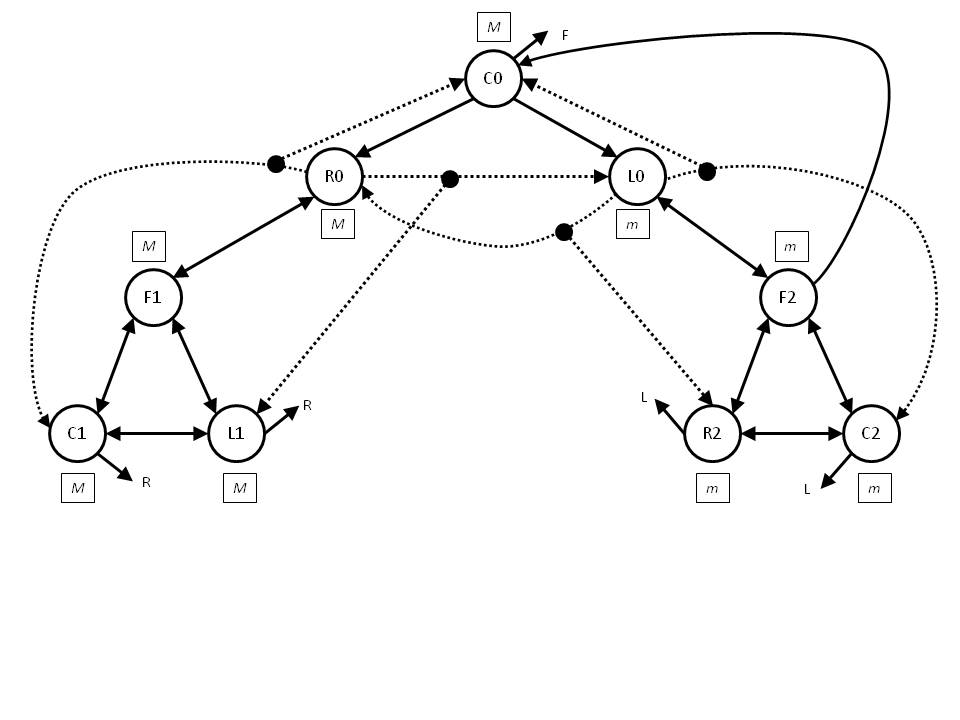}
\caption{Network $\sn$}
\label{sm_1-1}
\end{figure}

~

The communication links in $\sn$ are derived using the communication
graph $G(\sv,\se)$. The figure shows solid edges and dotted edges,
and also edges that do not terminate on one end.
We describe all three types of edges below.
\begin{itemize}
\item {\em Solid edges:}
If a node $i$ has a link to node $j$
in $G(\sv,\se)$, i.e., $(i, j) \in \se$,
then each copy of node $j$ in $\sn$ will have a link from
{\bf one of the copies} of node $i$ in $\sn$.
Exactly which copy of node $i$ has link to a copy of
node $j$ is represented with the edges shown in 
Figure \ref{sm_1-1}.
For instance,
the directed edge from vertex R0 to vertex F1 in Figure \ref{sm_1-1} indicates
that, {\bf if} for $r\in R$ and $k\in F$, link $(r,k)\in\se$,
then there is a link in $\sn$ from the copy of $r$ in R0 to
the copy of $k$ in F1.
Similarly,
the directed edge from vertex F2 to vertex L0 in Figure \ref{sm_1-1} indicates
that, {\bf if} for $k\in F$ and $l\in L$, link $(k,l)\in\se$,
then there is a link from the copy of $k$ in F2 to the copy of $l$ in L0.
Other solid edges in Figure \ref{sm_1-1} represent other communication
links in $\sn$ similarly.

\item {\em Dotted edges:}
Dotted edges are defined similar to the solid edges, with the difference being that
the dotted edges emulate a broadcast operation. 
Specifically, in certain cases, if link $(i,j)\in\se$, then one copy of node $i$
in $\sn$ may have links to {\bf two} copies of node $j$ in $\sn$, with both copies
of node $j$ receiving identical messages from the same copy of node $i$. 
This should be viewed as a ``broadcast'' operation that is being
emulated {\em unbeknownst} to the
nodes in $\sn$.
There are four such ``broadcast edges'' in the figure, shown as
dotted edges.
The broadcast edge from L0 to R0 and R2 indicates that
{\bf if} for $l\in L$ and $r\in R$, link $(l,r)\in\se$,
then messages from the copy of node $l$ in L0 are broadcast to the
copies of node $r$ in R0 and R2 both.
Similarly, the broadcast edge from R0 to C0 and C1 indicates that
{\bf if} for $r\in R$ and $c\in C$, link $(r,c)\in\se$,
then messages from the copy of node $r$ in R0 are broadcast to the
copies of node $c$ in C0 and C1 both.
There is also a broadcast edge from L0 to C0 and C2,
and another broadcast edge from R0 to L0 and L1. 

\item {\em ``Hanging'' edges:}
Five of the edges in Figure \ref{sm_1-1} do not terminate at any vertex.
One such edge originates at each of the vertices C1, L1, R2, C2, and C0,
and each such edge is labeled as R, L or F, as explained next.
A {\em hanging} edge signifies that the corresponding transmissions are discarded
silently without the knowledge of the sender.
In particular, the {\em hanging} edge originating at L1 with label R indicates the following:
if for $l\in L$ and $r\in R$, $(r,l)\in\se$, then transmissions by
the copy of node $l$ in L1 to node $r$ are silently discarded {\em without the knowledge} of
the copy of node $l$ in L1.
Similarly, the {\em hanging} edge originating at C0 with label F indicates the following:
if for $c\in C$ and $k\in F$, $(c,k)\in\se$, then transmissions by
the copy of node $c$ in C0 to node $k$ are silently discarded {\em without the knowledge} of
the copy of node $c$ in C0.

It is possible to avoid using such ``hanging'' edges by introducing additional vertices
in $\sn$. We choose the above approach to make the representation more compact.

\end{itemize}
Whenever $(i,j)\in\se$, in network $\sn$, each copy of node $j$ has
an incoming edge from {\bf one copy of} node $i$, as discussed above.
The broadcast and hanging edges defined above are consistent with our
{\em communication model} in Section \ref{s_intro}. As noted there,
each node, when sending a message, simply puts the message in the send
buffer. Thus, it is possible for us to emulate hanging edges by
discarding messages from the buffer, or broadcast edges by replicating
the messages into two send buffers. (Nodes do not read messages in
send buffers.) 

Now, let us assign input of $m$ or $M$, where $m\neq M$,
to each of the nodes in $\sn$. The inputs are shown next to the
vertices in small rectangles Figure \ref{sm_1-1}.
For instance, $M$ next to vertex C1 means that each node
represented by C1 has input $M$ (recall that C1 represents
one copy of each node in $C$). 
Similarly, $m$ next to vertex L0 means that each node
represented by L0 has input $m$.


Let $\beta$ denote a particular execution of ALGO in $\sn$ given the input specified above. Now, we identify three executions of ALGO in $G(\sv, \se)$ with a different set of nodes of size $\leq f$ behaving faulty. The behavior of the nodes is modeled by the corresponding nodes in $\sn$.

\begin{itemize}
\item {\bf Execution $\alpha_1$:}

Consider an execution $\alpha_1$ of ALGO in $G(\sv, \se)$, where the incoming neighbors of nodes in $R$ that are in $L$ or $C$, i.e., nodes in $N_R$, are faulty, with the rest of the nodes being fault-free. In addition, all the fault-free nodes
have inputs $M$. Now, we describe the behavior of each node.

\begin{itemize}
\item The behavior of fault-free nodes in $R$, $F$, $C-N_R$ and $L-N_R$ is modeled by the corresponding nodes in R0, F1, C1, and L1 in $\sn$. For example, nodes in $F$ send to their outgoing neighbors in $L$ the messages sent in $\beta$ by corresponding nodes in F1 to their outgoing neighbors in L1.

\item The behavior of the faulty nodes (i.e., nodes in $N_R$) is modeled by the behavior of the senders for the incoming links at the nodes in R0. In other words, faulty nodes are sending to their outgoing neighbors in $R$ the messages sent in $\beta$ by corresponding nodes in C0 or L0 to their outgoing neighbors in R0.
\end{itemize}
Recall from (\ref{e_nr}) that $|N_R|\leq f$. Since ALGO is correct in $G(\sv,\se)$, the nodes in $R$ must agree on $M$, because all the fault-free nodes in network I have input $M$.

\item {\bf Execution $\alpha_2$:}

Consider an execution $\alpha_2$ of ALGO in $G(\sv, \se)$, where the incoming neighbors of nodes in $L$ that are in $R$ or $C$, i.e., nodes in $N_L$, are faulty, with the rest of the nodes being fault-free. In addition, all the fault-free
nodes have inputs $m$. Now, we describe the behavior of each node.

\begin{itemize}
\item The behavior of fault-free nodes in $R-N_L$, $F$, $C-N_L$, and $L$ is modeled by the corresponding nodes in R2, F2, C2, and L0 in $\sn$. For example, nodes in $F$ send to their outgoing neighbors in $R$ the messages sent in $\beta$ by corresponding nodes in F2 to their outgoing neighbors in R2.

\item The behavior of the faulty nodes (i.e., nodes in $N_L$) is modeled by the behavior of the senders for the incoming links at the nodes in L0. In other words, faulty nodes are sending to their outgoing neighbors in $L$ the messages sent in $\beta$ by corresponding nodes in C0 or R0 to their outgoing neighbors in L0.
\end{itemize}
Recall from (\ref{e_nl}) that $|N_L|\leq f$. Since ALGO is correct in $G(\sv,\se)$, the nodes in $L$ must agree on $m$, because all the fault-free nodes in network II have input $m$.

\item {\bf Execution $\alpha_3$:}

Consider an execution $\alpha_3$ of ALGO in $G(\sv, \se)$, where the nodes in F are faulty, with the rest of the nodes being fault-free. In addition, nodes in $R \cup C$ have inputs $M$, and the nodes in $L$ have inputs $m$. Now, we describe the behavior of each node.

\begin{itemize}
\item The behavior of fault-free nodes in $R$, $C$ and $L$ is modeled by the corresponding nodes in R0, C0, and L0 in $\sn$. For example, nodes in $R$ are sending to their outgoing neighbors in $L$ the messages sent in $\beta$ by corresponding nodes in R0 to their outgoing neighbors in L0.

\item The behavior of the faulty nodes (i.e., nodes in $F$) is modeled by the nodes in F1 and F2. In particular, faulty nodes in $F$ send to their outgoing neighbors in $R$ the messages sent in $\beta$ by corresponding nodes in F1 to their outgoing neighbors in R0. Similarly the faulty nodes in $F$ send to their outgoing neighbors in $L$ the messages sent in $\beta$ by corresponding nodes in F2 to their outgoing neighbors in L0.  
\end{itemize}

Then we make the following two observations regarding $\alpha_3$:

\begin{itemize}
\item Nodes in $R$ must decide $M$ in $\alpha_3$
because, by construction, nodes in $R$ cannot distinguish between $\alpha_1$ and $\alpha_3$. 
Recall that nodes in $R$ decide on $M$ in $\alpha_1$.

\item Nodes in $L$ must decide $m$.
because by construction, nodes in $L$ cannot distinguish between $\alpha_2$ and $\alpha_3$.
Recall that nodes in $L$ decide on $m$ in $\alpha_2$.

\end{itemize}
Thus, in $\alpha_3$, the fault-free nodes in $R$ and $L$ decide on different
values, even though $|F|\leq f$. This violates the agreement condition,
contradicting the assumption that ALGO
is correct in $G(\sv,\se)$.

\end{itemize}

\end{proof}

\subsection{Equivalence of the Conditions Stated in Theorem \ref{t_nec_2} and Condition 1}

In this section, we first prove that Condition 1 (the condition in Theorem \ref{t_nec_1}) implies the condition in Theorem \ref{t_nec_2},
and then prove that the condition in Theorem \ref{t_nec_2} implies Condition 1 (the condition in Theorem \ref{t_nec_1}). Thus, the two conditions are proved to be equivalent.
We first prove the two lemmas below. The proofs use the following version of Menger's theorem \cite{Graph_theory_west}.

\begin{theorem}[Menger's Theorem]
Given a graph $G(\scriptv,\scripte)$ and two nodes $x, y \in \sv$, then a set $S \subseteq \sv - \{x,y\}$ is an $(x,y)$-cut if there is no $(x,y)$-path excluding $S$, i.e., every path from $x$ to $y$ must contain some nodes in $S$.
\end{theorem}

\begin{lemma}
\label{lemma:prop}
Assume that Condition 1 (the condition in Theorem \ref{t_nec_1}) holds for $G(\sv,\se)$.
For any partition
$A, B, F$ of $\scriptv$, where $A$
is non-empty, and $|F| \leq f$, if $B \not\Zightarrow A$,
then $\propagate{A}{B}{\sv-F}$.
\end{lemma}

\begin{proof}
Suppose that $A,B,F$ is a partition of $\sv$, where
$A$ is non-empty, $|F|\leq f$, and $B\not\rightarrow A$.
If $B=\emptyset$, then by Definition \ref{def:propagate},
the lemma is trivially true.
In the rest of this proof, assume that $B\neq\emptyset$.

%

Add a new (virtual) node $v$ to graph $G$, such that, (i) $v$ has no
incoming edges, (ii) $v$ has an outgoing edge to each node
in $A$, and (iii) $v$ has no outgoing edges to any node that is not in $A$.
Let $G_{+v}$ denote the graph resulting after the addition of $v$ to $G(\sv,\se)$
as described above. 

We want to prove that $\propagate{A}{B}{\sv-F}$.
Equivalently,\footnote{\label{f_prop} {\tt Footnote:} {\em Justification}:
Suppose that $\propagate{A}{B}{\sv-F}$. By the definition of $\propagate{A}{B}{\sv-F}$,
for each $b\in B$,
there exist at least $f+1$ disjoint $(A,b)$-paths excluding $F$;  these
paths only share
node $b$. Since $v$ has outgoing links to all the nodes in $A$, this implies
that there exist $f+1$ disjoint $(v,b)$-paths excluding $F$ in $G_{+v}$;
these paths only share
nodes $v$ and $b$.
Now, let us prove the converse. Suppose that there exist
$f+1$ disjoint $(v,b)$-paths excluding $F$ in $G_{+v}$. Node $v$ has
outgoing links
only to the nodes in $A$, therefore, from the $(f+1)$ disjoint $(v,b)$-paths
excluding $F$, if we
delete node $v$ and its outgoing links, then the shortened paths
are disjoint ($A,b$)-paths excluding $F$.}
we want to prove that, in graph $G_{+v}$, for each $b\in B$, there exist
$f+1$ disjoint ($v,b$)-paths excluding $F$.
We will prove this claim by contradiction.

Suppose that $\notpropagate{A}{B}{\sv-F}$, and
therefore, there exists a node $b\in B$ such that there are at most $f$
disjoint $(v,b)$ paths excluding $F$ in $G_{+v}$. By construction, there is no direct edge from $v$ to $b$. Then Menger's theorem \cite{Graph_theory_west} implies that
there exists a set $F_1\subseteq (A\cup B)-\{b\}$ with
$|F_1|\leq f$, such that, in graph $G_{+v}$, there is no $(v,b)$-path
excluding $F\cup F_1$. In other words, all $(v,b)$-paths
excluding $F$ contain at least one node in $F_1$.

Let us define the following sets $L,R,C$.
Some of the sets defined in this proof are illustrated in Figure \ref{fig:lemma6}.
\begin{itemize}
\item $L=A$.
	
	$L$ is non-empty, because $A$ is non-empty.
\item $R~=~\{~ i ~ | ~ i\in B-F_1~ \mbox{and 
there exists an ($i,b$)-path excluding $F\cup F_1$}\} $.

Thus, $R\subseteq B-F_1\subseteq B$.\\
Note that $b\in R$. Thus, $R$ is non-empty.

\item $C=B-R$. 

Thus, $C\subseteq B$. Since $R\subseteq B$, it follows
	that $R\cup C=B$.
 
\end{itemize}

\begin{figure}[tbhp]
\centering
\includegraphics[scale=0.6, bb=-30 -30 928 316]{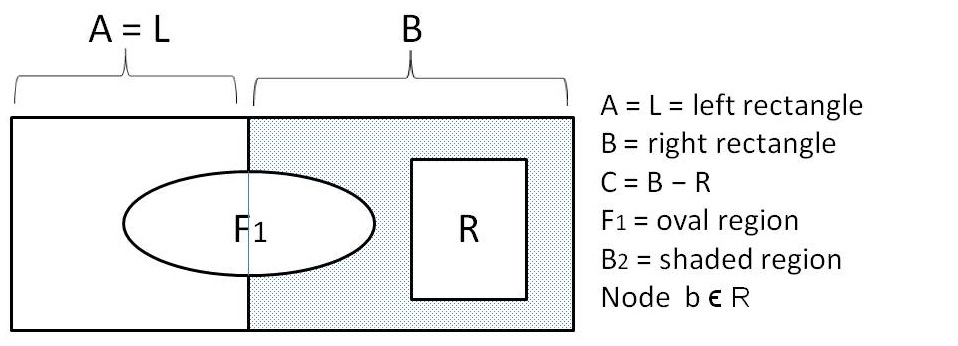}
\caption{Illustration for the proof of Lemma \ref{lemma:prop}}
\label{fig:lemma6}
\end{figure}

~

Observe that $L,R,C$ are disjoint sets, because $A$ and $B$ are disjoint, and
$L\cup R\cup C = A\cup B$.
Since set $F_1\subseteq A\cup B$, $L=A$, and $R\cap F_1=\emptyset$,
we have
$F_1\subseteq L\cup C$, and $F_1\cap B\subseteq C$.
Thus, set $C$ can be partitioned into disjoint sets $B_1$ and $B_2$
such that
\begin{itemize}
\item $B_1=C\cap F_1=B\cap F_1\subseteq C\subseteq B$, and
\item $B_2=C-B_1\subseteq C\subseteq B$. Note that $B_2\cap F_1=\emptyset$.
\end{itemize}
We make the following observations:
\begin{itemize}
\item For any $x\in A-F_1=L-F_1$ and $y\in R$, $(x,y)\not\in \se$.

	{\em Justification}:
	Recall that virtual node $v$ has a directed edge to $x$.
	If edge $(x,y)$ were to exist then there would be a $(v,b$)-path
	via nodes $x$ and $y$ excluding $F\cup F_1$
	(recall from definition of $R$ that $y$ has
	a path to $b$ excluding $F\cup F_1$).
	This contradicts the definition of set $F_1$.
\item For any $p\in B_2$, and $q\in R$, $(p,q)\not\in \se$.
	
	{\em Justification}:
	If edge $(p,q)$ were to exist, then there would be a $(p,b)$-path
	via node $q$ excluding $F\cup F_1$,
        since $q$ has a $(q,b)$-path excluding $F\cup F_1$.
	Then node $p$ should have been in $R$ by the definition of $R$.
	This is a contradiction to the assumption that $p\in B_2$, since
	$B_2\cap R\subseteq C\cap R=\emptyset$.
\end{itemize}
Thus, all the incoming neighbors of set $R$ are contained in $F \cup F_1$ (note that
$F_1=(A\cap F_1)\cup B_1$).
Recall that $F_1\subseteq L\cup C$.
Since $|F_1|\leq f$, it follows that
\begin{eqnarray}
L\cup C\not\Zightarrow R
\label{e1}
\end{eqnarray}

Recall that $B\not\Zightarrow A$.
By definitions of $L,R,C$ above, we have $A=L$ and $B=C\cup R$.
Thus,
\begin{eqnarray}
C\cup R\not\Zightarrow L
\label{e2}
\end{eqnarray}
(\ref{e1}) and (\ref{e2}) contradict the condition in Theorem \ref{t_nec_1}.
Thus, we have proved that
$\propagate{A}{B}{\sv-F}$.
\end{proof}

\begin{lemma}
\label{lemma:prop2}
Assume that Condition 1 (the condition in Theorem \ref{t_nec_1}) holds for $G(\sv,\se)$.
Consider a partition
$A, B, F$ of $\scriptv$, where $A, B$
are both non-empty, and $|F| \leq f$.
If $\notpropagate{B}{A}{\sv-F}$ then there exist $A'$ and $B'$ such
\begin{itemize}
\item $A'$ and $B'$ are both non-empty,
\item $A'$ and $B'$ form a partition of $A\cup B$,
\item $A'\subseteq A$ and $B\subseteq B'$, and
\item $ B'\not\Zightarrow A'$.
\end{itemize}
\end{lemma}
\begin{proof}
Suppose that 
$\notpropagate{B}{A}{\sv-F}$.

Add a new (virtual) node $w$ to graph $G$, such that, (i) $w$ has no
incoming edges, (ii) $w$ has an outgoing edge to each node
in $B$, and (iii) $w$ has no outgoing edges to any node that is not in $B$.
Let $G_{+w}$ denote the graph resulting after addition of $w$ to $G(\sv,\se)$
as described above.

Since $\notpropagate{B}{A}{\sv-F}$, for some node $a \in A$
there exist at most $f$ disjoint $(B,a)$-paths excluding $F$.
Therefore, there exist at most $f$ disjoint $(w,a)$-paths excluding
$F$ in $G_{+w}$.\footnote{See footnote \ref{f_prop}.}
Also by construction, $(w, a) \not\in \se$. Then, by Menger's theorem \cite{Graph_theory_west}, there must exist
$F_1\subseteq (A\cup B)-\{a\}$, $|F_1|\leq f$,
such that, in graph $G_{+w}$, all ($w,a$)-paths excluding $F$ 
contain at least one node in $F_1$.

Define the following sets (also recall that $\sv-F=A\cup B$):
\begin{itemize}
\item
$L ~=~\{~i~|~i\in \sv-F-F_1 ~\mbox{and there exists
	an $(i,a)$-path excluding $F\cup F_1$}~\}.$

\item
$R ~=~\{~j~|~j\in \sv-F-F_1 ~\mbox{and there exists in $G_{+w}$
	a $(w,j)$-path excluding $F\cup F_1$}~\}.$

Set $R$ contains $B-F_1$ since all nodes in $B$ have edges from $w$.

\item $C = \sv-F-L-R = (A\cup B) - L - R$.\\ Observe that $F_1\subseteq C$ (because
	nodes of $F_1$ are not in $L\cup R$). Also, by definition of
	$C$, sets $C$ and $L\cup R$ are disjoint.
\end{itemize}
Observe the following:
\begin{itemize}
\item
Sets $L$ and $R$ are disjoint, and set $L\subseteq A-F_1\subseteq A$.
Also, $A\cup B=L\cup R\cup C$.

{\em Justification}: $F_1\cap L=F_1\cap R=\emptyset$. By definition of $F_1$,
	all $(w,a)$-paths excluding $F$ contain at least one node in $F_1$.
If $L\cap R$ were to be non-empty, we can find a $(w,a)$-path excluding
$F\cup F_1$, which is a contradiction.

Note that $\sv-F-F_1=(A\cup B)-F_1$; therefore, $L\subseteq (A\cup B)-F_1$.
$B-F_1\subseteq R$, since all nodes in $B-F_1$ have links from $w$.
Since $L$ and $R$ are disjoint, it follows that $(B-F_1)\cap L=\emptyset$,
and therefore, $(A-F_1)\cap L=L$; that is, $L\subseteq A-F_1\subseteq A$.

%

\item For any $x\in C-F_1$ and $y\in L$, $(x,y)\not\in\se$.

{\em Justification}: If such a link were to exist, then
	$x$ should be in $L$, which is a contradiction
	(since $C$ and $L$ are disjoint).

\item There are no links from nodes in $R$ to nodes in $L$.

{\em Justification}: If such a link were to exist,
	it would contradict the definition of $F_1$, since we
	can now find a $(w,a)$-path excluding $F\cup F_1$.
\end{itemize}
Thus, all the incoming neighbors of set $L$ must be contained in
$F \cup F_1$. Recall that $F_1\subseteq C$ and $|F_1|\leq f$. Thus, 
\begin{eqnarray}
\label{e3}
R\cup C \not\Zightarrow L
\end{eqnarray}
Now define, $A' = L$,  $B' = R\cup C$.
Observe the following:
\begin{itemize}
\item $A'$ and $B'$ form a partition of $A\cup B$.

{\em Justification}:
$L,R,C$ are disjoint sets, therefore $A'=L$ and
$B'=R\cup C$ are disjoint. By the definition of
sets $L,R,C$ it follows that
$A'\cup B' = L\cup(R \cup C)=\sv-F=A\cup B$.

\item $A'$ is non-empty and $A'\subseteq A$.

{\em Justification}:
By definition of set $L$, set $L$ contains node $a$.
Thus, $A'=L$ is non-empty. We have already argued that $L\subseteq A$.
Thus, $A'\subseteq A$.


\item $B'$ is non-empty and $B\subseteq B'$.

{\em Justification}:
Recall that $L,R,C$ are disjoint,
and $L\cup R\cup C=A\cup B$. Thus, by definition of $C$,
$R\cup C=(A\cup B)-L$. Since $L\subseteq A$, it follows that
$B\subseteq R\cup C=B'$.
Also, since $B$ is non-empty, $B'$ is also non-empty.

\item $B' \not\Zightarrow A'$

{\em Justification}: Follows directly from (\ref{e3}),
and the definition of $A'$ and $B'$. 
\end{itemize}
This concludes the proof.
\end{proof}

~


\paragraph{Necessity Proof of Theorem \ref{t_nec_2}}

We now prove that Condition 1 (the condition in Theorem \ref{t_nec_1}) implies the condition in Theorem \ref{t_nec_2}.

\begin{proof}

Assume that Condition 1 (the condition in Theorem \ref{t_nec_1}) is satisfied by graph
$G(\sv,\se)$. Consider a partition of $A,B,F$ of $\sv$ such that
$A,B$ are non-empty and $|F|\leq f$.
Then, we must show that either $\propagate{A}{B}{\sv-F}$
or $\propagate{B}{A}{\sv-F}$.

\noindent
Consider two possibilities:
\begin{itemize}
\item $\propagate{B}{A}{\sv-F}$: In this case, the proof is complete.
\item $\notpropagate{B}{A}{\sv-F}$:
Then by Lemma \ref{lemma:prop2} in Appendix \ref{a_1and2}, there exist non-empty sets $A',B'$
that form a partition of $A\cup B$ such that
$A'\subseteq A$, $B\subseteq B'$, and $B'\not\Zightarrow A'$.
Lemma \ref{lemma:prop} in Appendix \ref{a_1and2} then implies that $\propagate{A'}{B'}{\sv-F}$.


Because $\propagate{A'}{B'}{\sv-F}$,
for each $b\in B'$, there exist $f+1$ disjoint $(A',b)$-paths excluding $F$.
Since $B\subseteq B'$, it then follows that,
for each $b\in B\subseteq B'$, there exist $f+1$ disjoint $(A',b)$-paths excluding $F$.
Since $A'\subseteq A$, and $F\cap A=\emptyset$, each $(A',b)$-path excluding $F$ is also a
$(A,b)$-path excluding $F$. Thus, 
for each $b\in B$, there exist $f+1$ disjoint $(A,b)$-paths excluding $F$.
Therefore, $\propagate{A}{B}{\sv-F}$.

\end{itemize}
\end{proof}

The proof above shows that
the {\em Condition 1} implies the condition
in Theorem \ref{t_nec_2}. That fact, and Lemma \ref{lemma:nec2-to-1} below,
together prove that the two forms of the condition are equivalent. Therefore, by Theorem \ref{t_nec_1}, the condition in Theorem \ref{t_nec_2} is necessary.

\begin{lemma}
\label{lemma:nec2-to-1}
The condition stated in Theorem \ref{t_nec_2} implies the condition stated in Theorem \ref{t_nec_1} (i.e., {\em Condition 1}).
\end{lemma}
\begin{proof}
We will prove the lemma by showing that, if {\em Condition 1} is violated, then the condition stated in Theorem \ref{t_nec_2}
is violated as well. 

Suppose that the {\em Condition 1} is violated. Then there exists a partition $L,C,R,F$ of $\sv$ such that $L,R$ are both non-empty, $|F|\leq f$, $L\cup C\not\Zightarrow R$ and $R\cup C\not\Zightarrow L$.

Since $L\cup C\not\Zightarrow R$, for any node $r\in R$, there exists a set $F_r$, $|F_r|\leq f$, such that all the $(L\cup C,r)$-paths excluding $F$ contain at least one node in $F_r$. Since $L\subseteq L\cup C$, Menger's theorem \cite{Graph_theory_west} implies that there are at most $f$ disjoint  $(L,r)$-paths excluding $F$.  Thus, because $r\in R\cup C$, $\notpropagate{L}{R\cup C}{\sv-F}$. 

Similarly, since $R\cup C\not\Zightarrow L$, for any node $l\in L$, there exists a set $F_l$, $|F_l|\leq f$, such that all the $(R\cup C,l)$-paths excluding $F$  contain at least one node in $F_l$. Menger's theorem \cite{Graph_theory_west} then implies that there are at most $f$ disjoint $(R\cup C,l)$-paths excluding $F$.  Thus, $\notpropagate{R\cup C}{L}{\sv-F}$.

Define $A=L$, and $B=R\cup C$. Thus, $A,B,F$ is a partition of $\sv$ such that $|F|\leq f$ and $A,B$ are non-empty. The two conditions derived above imply that
$\notpropagate{A}{B}{\sv-F}$ and $\notpropagate{B}{A}{\sv-F}$,  violating the condition stated in Theorem \ref{t_nec_2}.
\end{proof}

\section{2-clique Network}
\label{a_l_2clique}

In this section, we present a family of graphs, namely 2-clique network. We will prove that the graph satisfies the necessary condition in Theorem \ref{t_nec_2}, but each pair of nodes may not be able to communicate reliably with each other. 

\begin{definition}
\label{def:dual}
A graph $G(\sv, \se)$ consisting of $n = 6f+2$ nodes, where $f$ is a positive even integer, is said to be a {\em 2-clique network} if all the following properties are satisfied: 

\begin{itemize}
\item It includes two disjoint cliques, each consisting of $3f+1$ nodes.
Suppose that the nodes in the two cliques are specified by sets $K_1$ and $K_2$,
respectively, where $K_1=\{u_1,u_2,\cdots, u_{3f+1}\}\subset \scriptv$,
and $K_2=\scriptv-K_1=\{w_1,w_2,\cdots, w_{3f+1}\}$.
Thus, $(u_i,u_j)\in\scripte$ and $(w_i,w_j)\in\scripte$, for $1\leq i,j\leq 3f+1$ and $i \neq j$,
\item $(u_i,w_i)\in\scripte$, for $1\leq i\leq \frac{3f}{2}$ 
		and $i=3f+1$, and
\item $(w_i,u_i)\in\scripte$, for $\frac{3f}{2}+1\leq i\leq 3f$ 
		and $i=3f+1$.
\end{itemize}
\end{definition}
Figure \ref{f:2-core} is the 2-clique network for $f=2$. Note that Section \ref{s_sufficiency} proves that Byzantine consensus is possible in all graphs that satisfy the necessary condition. Therefore, consensus is possible in the 2-clique network as well. 

~

We first prove the following lemma for any graph $G(\sv, \se)$ that satisfies the necessary
condition.

\begin{lemma}
\label{lemma:AtoBC}
Let $A, B, C, F$ be disjoint subsets of $\sv$ such that $|F| \leq f$ and $A, B, C$ are non-empty. Suppose that $\propagate{A}{B}{\sv-F}$ and $\propagate{A \cup B}{C}{\sv-F}$. Then, $\propagate{A}{B\cup C}{\sv-F}$.
\end{lemma}

\begin{proof}
The proof is by contradiction. Suppose that 

\begin{itemize}
\item $\propagate{A}{B}{\sv-F}$, 
\item $\propagate{A \cup B}{C}{\sv-F}$, and
\item $\notpropagate{A}{B \cup C}{\sv-F}$.

\end{itemize}
The first condition above implies that $|A|\geq f+1$.
By Definition \ref{def:propagate} and Menger's Theorem \cite{Graph_theory_west}, the third condition implies that there exists a node $v \in B\cup C$ and a set of nodes $P \subseteq \sv-F-\{v\}$ such that $|P| \leq f$, and all $(A, v)$-paths excluding $F$ contain at least one node in $P$.
In other words, there is no $(A,v)$-path excluding $F\cup P$.
Observe that, because
$\propagate{A}{B}{\sv-F}$, $v$ cannot be in $B$; therefore
$v$ must belong to set $C$.

Let us define the sets $X$ and $Y$ as follows:
\begin{itemize}
\item Node $x\in X$ if and only if $x\in \scriptv-F-P$
	and there exists an $(A,x)$-path
	excluding $F\cup P$.
	It is possible that $P\cap A\neq\emptyset$; thus, the $(A,x)$-path
	cannot contain any nodes in $P\cap A$. 
\item Node $y\in Y$ if and only if  $y\in\scriptv-F-P$ and there exists an $(y,v)$-path
	excluding $F\cup P$.
\end{itemize}

By the definition of $X$ and $Y$, it follows that
for any $x\in X,~y\in Y$, there cannot be any
$(x,y)$-path excluding $F\cup P$.
Also, since $\propagate{A}{B}{\scriptv-F}$, for each $b\in B-P$,
there must exist an $(A,b)$-path excluding $F\cup P$; thus,
$B-P\subseteq X$, and $B\subseteq X\cup P$.
Similarly, $A \subseteq X\cup P$, and therefore, $A\cup B\subseteq X\cup P$.

By definition of $X$, there are no $(X\cup P, v)$-paths excluding
$F\cup P$. Therefore, because $A\cup B\subseteq X\cup P$, there are
no $(A\cup B,v)$-paths excluding $F\cup P$. Therefore,
since $v\in C$, $\notpropagate{A\cup B}{C}{\scriptv-F}$. This is a contradiction to the second condition above.
\end{proof}

~

Now, we use Lemma \ref{lemma:AtoBC} to prove the following Lemma.

\begin{lemma}
\label{l_2clique}
Suppose that $G(\scriptv,\scripte)$ is a 2-clique network.
Then graph $G$ satisfies the condition in Theorem \ref{t_nec_2}.
\end{lemma}

\begin{proof}


Consider a partition $A, B, F$ of $\sv$, where $A$ and $B$ are both non-empty, and $|F| \leq f$. Recall from Definition \ref{def:dual} that $K_1,K_2$ also
form a partition of $\sv$.

Define $A_1 = A \cap K_1, A_2 = A \cap K_2, B_1 = B \cap K_1,B_2 = B \cap K_2, F_1=F\cap K_1$ and $F_2=F\cap K_2$.

Define $\scripte'$ to be the set of directed links from the nodes
in $K_1$ to the nodes in $K_2$, or vice-versa.
Thus, there are
$\frac{3f}{2}+1$ directed links in $\scripte'$
from the nodes in $K_1$ to the nodes in $K_2$,
and the same number of links from the nodes in $K_2$ to the nodes
in $K_1$. 
Each pair of links in $\scripte'$, with the exception of
the link pair between $a_{3f+1}$ and $b_{3f+1}$, is node disjoint.
Since $|F|\leq f$, it should be easy to see that, at least
one of the two conditions below is true:
\begin{list}{}{}
\item{(a)} There are at least $f+1$ directed links from the nodes in $K_1-F$ to the nodes in $K_2-F$. 
\item{(b)} There are at least $f+1$ directed links from the nodes in $K_2-F$ to nodes the in $K_1-F$. 
\end{list}
Without loss of generality, suppose that condition (a) is true.
Therefore, since $|K_1-F|\geq 2f+1$
and the nodes in $K_2-F$ form a clique, it follows that $\propagate{K_1-F}{K_2-F}{\sv - F}$.
Then, because $K_1-F=A_1\cup B_1$ and $K_2-F=A_2\cup B_2$, 
we have
\begin{eqnarray}
\propagate{A_1\cup B_1}{A_2\cup B_2}{\scriptv-F}.
\label{e_topo_1}
\end{eqnarray}

$|K_1-F|\geq 2f+1$ also implies that
either $|A_1|\geq f+1$ or $|B_1|\geq f+1$.
Without loss of generality, suppose that $|A_1|\geq f+1$.
Then, since the nodes in $A_1\cup B_1$ form a clique, it follows that
$\propagate{A_1}{B_1}{\scriptv-F_1-K_2}$
(recall that $\scriptv-F_1-K_2=A_1\cup B_1$). Since $\scriptv-F_1 - K_2 \subset
\scriptv-F$, we have 
\begin{eqnarray}
\propagate{A_1}{B_1}{\scriptv-F}
\label{e_topo_2}
\end{eqnarray}

(\ref{e_topo_1}) and (\ref{e_topo_2}), along with 
Lemma \ref{lemma:AtoBC} above imply that
$\propagate{A_1}{B_1 \cup A_2 \cup B_2}{\sv-F}$. Therefore,
$\propagate{A_1}{B_1 \cup B_2}{\sv-F}$, and $\propagate{A_1 \cup A_2}{B_1 \cup B_2}{\sv-F}$. 
Since $A=A_1\cup A_2$ and $B=B_1\cup B_2$, $\propagate{A}{B}{\sv - F}$.
\end{proof}

\section{Proof of Corollary \ref{cor:2f+1}}
\label{a_cor:2f+1}

\begin{proof}
Since $n\geq 3f+1$ is a necessary condition for Byzantine consensus in
undirected graphs \cite{impossible_proof_lynch, welch_book}, it follows that $n\geq 3f+1$ is also necessary for directed graphs.  As presented below, this necessary condition can also be derived from Theorem \ref{t_nec_2}.

For $f=0$, condition (i) in the corollary is trivially true.
Now consider $f>0$.
 The proof is by contradiction. Suppose that $n\leq 3f$.
As stated in Section \ref{s_intro}, we assume $n\geq 2$, since consensus for $n=1$ is trivial.
Partition $\sv$ into three subsets $A,B,F$ such that $|F|\leq f$,
$0<|A|\leq f$, and $0<|B|\leq f$. Such a partition can be found because
$2\leq |\sv|\leq 3f$. 
Since $A,B$ are both non-empty, and contain at most $f$ nodes each, we have 
$\notpropagate{A}{B}{\sv-F}$ and $\notpropagate{B}{A}{\sv-F}$, violating the condition in Theorem \ref{t_nec_2}.
Thus, $n\geq 3f+1$ is a necessary condition.

Now, for $f>0$, we show that it is necessary for each node to have at least $2f+1$ incoming neighbors. The proof is by contradiction.
Suppose that for some node $i \in \sv$, the number of incoming neighbors
is at most $2f$. Partition $\sv-\{i\}$ into two sets $L$ and $F$
such that $L$ is non-empty and contains at most $f$ incoming neighbors of $i$,
and $|F|\leq f$. It should be easy to see that such $L,F$ can be found, since node $i$ has at most $2f$ incoming neighbors.

Define $A = \{i\}$ and $B = \sv - A - F = L$.
Thus, $A,B,F$ form a partition of $\sv$.
Then, since $f>0$ and $|A|=1, |B| = |L| > 0$, it follows that $\notpropagate{A}{B}{\sv-F}$.
Also, since $B$ contains at most $f$ incoming neighbors of node $i$, and
set $A$ contains only node $i$, there are at most $f$ node-disjoint $(B,i)$-paths. Thus, $\notpropagate{B}{A}{\sv-F}$. The above two conditions violate the necessary condition stated in Theorem \ref{t_nec_2}. 
\end{proof}

\section{Source Component}
\label{a_SC}

We introduce some definitions and results that are useful in the other appendices.

\begin{definition}
\label{def:decompose}
{\bf Graph decomposition:}
Let $H$ be a subgraph of $G(\sv,\se)$. Partition graph $H$ into non-empty strongly connected components,
 $H_1,H_2,\cdots,H_h$, where $h$ is a non-zero integer dependent on graph $H$,
 such that nodes $i,j\in H_k$ if and only if
there exist $(i,j)$- and $(j,i)$-paths both excluding nodes outside $H_k$.

Construct a graph $H^d$ wherein each strongly connected component $H_k$ above is represented by vertex $c_k$, and there is an edge from vertex $c_k$ to vertex $c_l$ if and only if the nodes in $H_k$ have directed paths in $H$ to the nodes in $H_l$.
\end{definition}
It is known that the decomposition graph $H^d$ is a directed {\em acyclic} graph \cite{dag_decomposition}.

\begin{definition}
\label{def:source_comp}
{\bf Source component}:
Let $H$ be a directed graph, and let $H^d$ be its decomposition as per Definition~\ref{def:decompose}. Strongly connected component $H_k$ of $H$ is said to be a {\em source component} if the corresponding vertex $c_k$ in $H^d$ is \underline{not} reachable from any other vertex in $H^d$. 
\end{definition}


\begin{definition}
\label{def:reduced} {\bf Reduced Graph:}
For a given graph $G(\scriptv,\scripte)$, and sets $F\subset\scriptv$,
$F_1\subset \sv-F$, such that $|F|\leq f$ and $|F_1|\leq f$, reduced
graph $G_{F,F_1}(\scriptv_{F,F_1},\scripte_{F,F_1})$ is defined as
follows: (i)
$\scriptv_{F,F_1}=\scriptv-F$, and (ii) $\scripte_{F,F_1}$ is obtained by removing from $\scripte$
all the links incident on the nodes in $F$, and all the outgoing links from nodes
in $F_1$. 
That is, $\se_{F,F_1}=\se-\{(i,j)~|~i\in F ~~\mbox{or} ~~ j\in F\}
			- \{(i,j)~|~i\in F_1\}$.
\end{definition}

\begin{corollary}
\label{cor:1:prop}
Suppose that graph $G(\scriptv,\scripte)$ satisfies the condition stated in
Theorem \ref{t_nec_1}. For any $F \subset \scriptv$
and $F_1\subset \sv-F$, such that $|F| \leq f$ and
$|F_1|\leq f$, let $S$ denote the set of nodes in the source component of
$G_{F,F_1}$.
Then, $\propagate{S}{\sv-F-S}{\sv-F}$.
\end{corollary}

\begin{proof}
Since $G_{F,F_1}$ contains non-zero number of nodes, its source component $S$ must be non-empty.
If $\sv-F-S$ is empty, then the corollary follows trivially by Definition \ref{def:propagate}.
Suppose that $\sv-F-S$ is non-empty.
Since $S$ is a source component in $G_{F,F_1}$, it has
no incoming neighbors in $G_{F,F_1}$; therefore, all of the incoming neighbors
of $S$ in $\sv-F$ in graph $G(\sv,\se)$ must belong to $F_1$. Since $|F_1|\leq f$, we have,
\[
(\scriptv - S - F) \not\Zightarrow S
\]
Lemma \ref{lemma:prop} in Appendix \ref{a_1and2} then implies that
\[
\propagate{S}{\sv-F-S}{\sv-F}
\]
\end{proof}

~

\begin{lemma}
\label{l_connected}
For any $F\subset\sv$, $F_1\subset\sv-F$, such that $|F|\leq f$, $|F_1|\leq f$: 
\begin{itemize}
\item The source component of $G_{F,F_1}$ is strongly connected in $G_{-F}$. ($G_{-F}$ is defined in Definition \ref{def:G-F} in Section \ref{s_sufficiency}.)
\item The source component of $G_{F,F_1}$ does not contain any nodes in $F_1$.
\end{itemize}
\end{lemma}

\begin{proof}
By Definition \ref{def:decompose}, each pair of nodes $i,j$ in the source
component of graph $G_{F,F_1}$ has at least one $(i,j)$-path 
and at least one $(j,i)$-path consisting of nodes only in $G_{F,F_1}$,
i.e., excluding nodes in $F$.

Since $F_1\subset \sv-F$, $G_{F,F_1}$ contains other nodes besides $F_1$.
Although nodes of $F_1$ belong to graph $G_{F,F_1}$, the
nodes in $F_1$ do not have any outgoing links in $G_{F,F_1}$.
Thus, a node in $F_1$ cannot have paths to any other node in $G_{F,F_1}$.
Then, due to the connetedness requirement of a source component,
it follows that no nodes of $F_1$ can be in the source component.
\end{proof}

\section{Sufficiency for $f=0$}
\label{a_f_0}

The proof below uses the terminologies and results presented in Appendix \ref{a_SC}.
We now prove that, when $f=0$, the necessary condition in Theorem \ref{t_nec_2} is
sufficient to achieve consensus.

\begin{proof}

When $f=0$,
suppose that the graph $G$ satisfies the necessary condition in Theorem \ref{t_nec_2}.
Consider the source component $S$ in reduced graph $G_{\emptyset, \emptyset} = G$, i.e., in the reduced graph where $F = F_1 = \emptyset$, as per Definition \ref{def:reduced} in Appendix \ref{a_SC}. Note that by definition, $S$ is non-empty. Pick a node $i$ in the source component. By Lemma \ref{l_connected} in Appendix \ref{a_SC}, $S$ is strongly connected in $G$, and thus $i$ has a directed path to each of the nodes in $S$. By Corollary \ref{cor:1:prop} in Appendix \ref{a_SC}, because $F=\emptyset$, $\propagate{S}{\sv-S}{\sv}$, i.e., for each node $j \in \sv-S$, an $(S, j)$-path exists. Since $S$ is strongly connected, an $(i, j)$-path also exists.  Then consensus can be achieved simply by node $i$ routing its input to all the other nodes, and requiring all the nodes to adopt node $i$'s input as the output (or decision) for the consensus.  It should be easy to see that termination, validity and agreement properties are all satisfied.

\end{proof}






\section{Proof of Claim \ref{claim:inner_loop} in Section \ref{subsec:inner}}
\label{a_claim:S}

The proof below uses the terminologies and results presented in Appendix \ref{a_SC}. We first prove a simple lemma.

\begin{lemma}
\label{l_propagate_f+1}
Given a partition $A,B,F$ of $\sv$ such that $B$ is non-empty,
and $|F|\leq f$,
if $\propagate{A}{B}{\sv-F}$, then
size of $A$ must be at least $f+1$.
\end{lemma}
\begin{proof}
By definition, there must be at least $f+1$ disjoint $(A,b)$-paths
excluding $F$ for each $b\in B$. Each of these $f+1$ disjoint paths
will have a distinct {\em source} node in $A$. Therefore, such $f+1$ disjoint paths
can only exist if $A$ contains at least $f+1$ distinct nodes.
\end{proof}

~

~

\noindent

We now prove the claim (i) in Section \ref{subsec:inner}. \\



\noindent{\bf Proof of Claim \ref{claim:inner_loop} in Section \ref{subsec:inner}:}

The second claim of Claim \ref{claim:inner_loop} is proved in the main body already. Now, we present the proof of the first claim:

\noindent
{\em The required set $S$ exists in both Case 1 and 2 of each INNER loop.}

~

Consider the two cases in the INNER loop.
\begin{itemize}
\item{Case 1:~}
 $\propagate{A}{B}{\sv-F}$ and $\notpropagate{B}{A}{\sv-F}$:

Since $\notpropagate{B}{A}{\sv-F}$, by Lemma \ref{lemma:prop2} in Appendix \ref{a_1and2},
there exist non-empty sets $A',B'$ that form a partition of
$A\cup B=\sv-F$  such that $A'\subseteq A$ and
\[
B'\not\Zightarrow A'
\]
Let $F_1$ be the set of incoming neighbors of $A'$ in $B'$.
Since $B'\not\Zightarrow A'$, $|F_1|\leq f$. Then
$A'$ has no incoming neighbors in $G_{F,F_1}$. Therefore, the
source component of $G_{F,F_1}$ 
must be contained within $A'$. (The definition of source component is in Appendix \ref{a_SC}.)
Let $S$ denote the set of nodes in this source component.
Since $S$ is the source component, 
by Corollary \ref{cor:1:prop} in Appendix \ref{a_SC},
\[
\propagate{S}{\sv-S-F}{\sv-F}.
\]
Since $S\subseteq A'$ and $A'\subseteq A$, $S\subseteq A$.
Then, $B\subseteq (A\cup B)-S =\sv-S-F$; therefore, $\sv-S-F$ is non-empty.
Also, since $\propagate{S}{\sv-S-F}{\sv-F}$, set $S$ must be non-empty
(by Lemma \ref{l_propagate_f+1} above).
By Lemma \ref{l_connected} in Appendix \ref{a_SC}, $S$ is strongly connected in $G_{-F}$. (The definition of $G_{-F}$ is in Section \ref{s_sufficiency}.)
Thus, set $S$ as required in Case 1 exists.

\item{Case 2:~}
 $\propagate{A}{B}{\sv-F}$ and $\propagate{B}{A}{\sv-F}$:

Recall that we consider $f>0$ in Section \ref{subsec:inner}. 


By Corollary \ref{cor:2f+1} in Section \ref{nec_2}, since $|\sv|=n>3f$, $|A\cup B|=|\sv-F|>2f$. 
In this case, we pick an arbitrary non-empty set $F_1\subset A\cup B= \sv-F$ such that $|F_1|=f>0$, and find the source component of $G_{F,F_1}$. 
Let the set of nodes in the source component be denoted as $S$.
Since $S$ is the source component, 
by Corollary \ref{cor:1:prop} in Appendix \ref{a_SC},
\[
\propagate{S}{\sv-F-S}{\sv-F}
\]
Also, since $\propagate{A}{B}{\sv-F}$, and $(S-A)\subseteq B$,
we have $\propagate{A}{(S-A)}{\sv-F}$.
Also, since $\sv-S-F$ contains $F_1$, and $F_1$ is non-empty,
$\sv-S-F$ is non-empty; also, since $\propagate{S}{\sv-S-F}{\sv-F}$,
set $S$ must be non-empty (by Lemma \ref{l_propagate_f+1} above).
By Lemma \ref{l_connected} in Appendix \ref{a_SC}, $S$ is strongly connected in $G_{-F}$.
Thus, set $S$ as required in Case 2 exists.
\end{itemize}


\fillbox

~

\comment{++++++
Consider nodes in set $F$. As shown in Corollary \ref{cor:2f+1} in Section \ref{nec_2},
when $f>0$, each node in $\sv$ has at least $2f+1$ incoming neighbors.
Since $|F|\leq f$,
for each $k\in F$ there must exist at least
$f+2$ incoming neighbors in $\sv-F$. 
Thus, the desired set $N_k$ exists, satisfying
the requirement in step (j) of Algorithm BC.

+++++++++++++ f+1??? ++++++++++++}

\fillbox

\section{Proof of Lemma \ref{l_agreement}}
\label{a_l_agreement}

The proof below uses the terminologies and results presented in Appendix \ref{a_SC}. Now, we present the proof of Lemma \ref{l_agreement}.

\begin{proof}
Recall that $F^*$ denotes the \underline{actual set of faulty nodes} in
the network ($0\leq |F^*|\leq f$).

Since the OUTER loop of Algorithm BC considers all possible $F\subseteq \sv$
such that $|F|\leq f$, eventually, the OUTER loop will be
performed with $F=F^*$.

In the INNER loop for $F=F^*$, different partitions $A,B$
of $\sv-F=\sv-F^*$ will be considered.
We will say that such a partition $A,B$ is a ``conformant'' partition
if $v_i=v_j$ for all $i,j\in A$, 
and $v_i=v_j$ for all $i,j\in B$.
A partition $A,B$ that is not conformant is said to be ``non-conformant''. Further, we will say that an INNER loop iteration is a ``deciding'' iteration if one of the following condition is true. 

\begin{itemize}
\item[C1]: The $A,B$ partition of $\sv-F$ considered in the iteration is conformant.

In Case 1 with conformant partition, every node in $S$ has the same value $t$  after step (a). Hence, in the end of step (b), every node in $S$ has the same value $t$. Now, consider Case 2 with conformant partition. Denote the value of all the nodes in $A$ by $\alpha$ ($\alpha \in \{0,1\}$). Then, in step (e), each node
 $i$ in $A$ (including $S \cap A$) sets $t_i$ equal to $\alpha$. In step (f), all the nodes in $S \cap B$ receive identical values $\alpha$ from nodes in $A$, and hence, they set value $t$  equal to $\alpha$. Therefore, every node in $S$ has the same value $t$ at the end of step (g).

\item[C2]: The $A,B$ partition of $\sv-F$ considered in the iteration is non-conformant; however,
the values at the nodes are such that, at the end of step (b) of Case 1, or at the end of step (g) of Case 2 (depending on which case applies), every node in the corresponding set $S$ has the same value $t$. (The definition of source component is in Appendix \ref{a_SC}.) That is, for all $i, j \in S,~ t_i = t_j$.
\end{itemize}
In both C1 and C2, all the nodes in the corresponding source component $S$ have the identical value $t$ in the deciding iteration (in the end of step (b) of Case 1, and in the end of step (g) of Case 2). The iteration that is not deciding is said to be ``non-deciding''. 

\begin{claim}
\label{claim:deciding}
In the INNER loop with $F=F^*$, value $v_i$ for each fault-free node $i$ will stay unchanged in every non-deciding iteration.
\end{claim}

\begin{proof}
Suppose that $F=F^*$, and the INNER loop iteration under consideration is a non-deciding iteration. Observe that since the paths used in procedures {\tt Equality} and {\tt Propagate} exclude $F$, none of the faulty nodes can affect the outcome of any INNER loop iteration when $F=F^*$. Thus, during \Equality($S$) (step (b) of Case 1, and step (g) of Case 2), each node in $S$ can receive the value from other nodes in $S$ correctly. Then, every node in $S$ will set value $t$ to be $\perp$ in the end of \Equality($S$), since by the definition of non-deciding iteration, there is a pair of nodes $j, k \in S$ such that $t_j \neq t_k$. Hence, every node in $\sv - F - S$ will receive $f+1$ copies of $\perp$ after \Propagate($S, \sv-F-S$) (step (c) of Case 1, and step (h) of Case 2), and will set value $t$ to $\perp$. Finally, at the end of the INNER loop iteration, the value $v$ at each node stays unchanged based on the following two observations:

\begin{itemize}
\item nodes in $S$ in Case 1, and in $A \cap S$ in Case 2, will not change value $v$ as specified by Algorithm BC, and

\item $t_i = \perp$ for each node $i \in \sv-F-S$ in Case 1, and for each node $i \in \sv-F-(A \cap S)$
in Case 2.
\end{itemize}
Thus, no node in $\sv-F$ will change their $v$ value (where $F=F^*$).

Note that by assumption, there is no fault-free node in $F=F^*$, and hence, we do not need to consider STEP 2 of the INNER loop. Therefore, Claim \ref{claim:deciding} is proved.
\end{proof}

~

Let us divide the INNER loop iterations for $F=F^*$ into three phases:
\begin{itemize}
\item Phase 1: INNER loop iterations before the first deciding iteration with $F=F^*$.

\item Phase 2: The first deciding iteration with $F=F^*$.

\item Phase 3: Remaining INNER loop iterations with $F=F^*$.

\end{itemize}

\begin{claim}
\label{claim:phase2}
At least one INNER loop iteration with $F=F^*$ is a deciding iteration.
\end{claim}

\begin{proof}
The input at each process is in $\{0,1\}$. Therefore,
by repeated application of Lemma \ref{l_valid_1} in Section \ref{s:correct}, it is always true that $v_i \in \{0,1\}$ for each fault-free node $i$. Thus, when the OUTER iteration for $F=F^*$ begins, a conformant partition exists (in particular, set $A$ containing all fault-free nodes with $v$ value $0$, and set $B$ containing the remaining fault-free nodes, or vice-versa.) 
By Claim \ref{claim:deciding}, nodes in $\sv - F$ will not change values during non-deciding iterations.
Then, since the INNER loop considers all partitions of $\sv-F$, the INNER loop will eventually consider
either the above conformant partition, or sometime prior to considering the above conformant
partition, it will consider a non-conformant partition with properties in (C2) above.
\end{proof}

~

Thus, Phase 2 will be eventually performed when $F=F^*$. Now, let us consider each phase separately:
\begin{itemize}

\item Phase 1: Recall that all the nodes in $\sv-F=\sv-F^*$ are fault-free. By Claim \ref{claim:deciding}, the $v_i$ at each fault-free node $i \in \sv-F$ stays unchanged.

\item Phase 2: Now, consider the first deciding iteration of the INNER loop. 

Recall from Algorithm BC that a suitable set $S$ is identified in each INNER loop iteration.
 We will show that in the deciding iteration, every node in $S$ will have the same $t$ value. Consider two scenarios:

\begin{itemize}
\item The partition is non-conformant: Then by definition of deciding iteration, we can find an $\alpha \in \{0,1\}$ such that $v_i=\alpha$ for all $i\in S$ after step (b) of Case 1, or after step (g) of Case 2.

\item The partition is conformant: Let $v_i = \alpha$ for all $i \in A$ for $\alpha \in \{0,1\}$. Such an $\alpha$ exists because the partition is conformant.

\begin{itemize}
\item Case 1: In this case, recall that $S \subseteq A$. Therefore, after steps (a) and (b) both, $t_j$ at all $j\in S$ will be identical, and equal to $\alpha$.

\item Case 2: This is similar to Case 1. At the end of step (e), for all nodes
$i\in A$, $t_i=\alpha$.
After step (f), for all nodes $i\in S\cup A$, $t_i=\alpha$.
Therefore, after step (g),  for all nodes $i\in S$, $t_i$ will remain equal to $\alpha$.
\end{itemize}

\end{itemize}

Thus, in both scenarios above, we found a set $S$ and $\alpha$ such that for all $i \in S$, $t_i = \alpha$ after step (b) in Case 1, and after step (g) in Case 2.

Then, consider the remaining steps in the deciding iteration.

\begin{itemize}
\item Case 1: During \Propagate($S,\sv-F-S$), each node $k\in\sv-F-S$ will receive $f+1$ copies of $\alpha$ along $f+1$ disjoint paths, and set $t_k=\alpha$ in step (c). Therefore, each node $k\in \sv-F-S$ will update its $v_k$ to be $\alpha$ in step (d).
(Each node $p\in S$ does not modify its $v_p$, which is
already equal to $\alpha$.)

\item Case 2: After step (h), $t_j=\alpha$ for all $j\in (\sv-F-S)\cup S$. Thus, each node  $k\in \sv-F-(A\cap S)$ will update $v_k$ to be $\alpha$.
(Each node $p\in A\cap S$ does not modify its $v_p$, which is
already equal to $\alpha$.)

\end{itemize}

Thus, in both cases, at the end of STEP 1 of the INNER loop, for
all $k\in\sv-F=\sv-F^*$, $v_k=\alpha$. 

Since all nodes in $F^*$ are faulty, agreement has been reached at this
point.
The goal now is to show that the agreement
property is not violated by
actions taken in any future INNER loop iterations.

\item Phase 3: At the start of Phase 3, for each fault-free 
	node $k\in \sv-F^*$, we have $v_k=\alpha\in\{0,1\}$.
	Then by Lemma \ref{l_valid_1}, all future INNER loop iterations cannot assign any value other than $\alpha$ to any node $k\in\sv-F^*$. 

\end{itemize}
After Phase 3 with $F=F^*$, Algorithm BC may perform OUTER loop iterations for
other choices of set $F$. However, due to Lemma \ref{l_valid_1},
the value $v_i$ at each $i\in\sv-F^*$ (i.e., all fault-free nodes)
continues being equal to $\alpha$. 

Thus, Algorithm BC satisfies the {\em agreement} property, as stated
in Section \ref{s_intro}.
\end{proof}


\begin{thebibliography}{10}

\bibitem{BFT-CUP_OPODIS}
E.~Alchieri, A.~Bessani, J.~Silva~Fraga, and F.~Greve.
\newblock Byzantine consensus with unknown participants.
\newblock In T.~Baker, A.~Bui, and S.~Tixeuil, editors, {\em Principles of
  Distributed Systems}, volume 5401 of {\em Lecture Notes in Computer Science},
  pages 22--40. Springer Berlin Heidelberg, 2008.

\bibitem{welch_book}
H.~Attiya and J.~Welch.
\newblock {\em Distributed Computing: Fundamentals, Simulations, and Advanced
  Topics}.
\newblock Wiley Series on Parallel and Distributed Computing, 2004.

\bibitem{Bansal_disc11}
P.~Bansal, P.~Gopal, A.~Gupta, K.~Srinathan, and P.~K. Vasishta.
\newblock Byzantine agreement using partial authentication.
\newblock In {\em Proceedings of the 25th international conference on
  Distributed computing}, DISC'11, pages 389--403, Berlin, Heidelberg, 2011.
  Springer-Verlag.

\bibitem{dag_decomposition}
S.~Dasgupta, C.~Papadimitriou, and U.~Vazirani.
\newblock {\em Algorithms}.
\newblock McGraw-Hill Higher Education, 2006.

\bibitem{yvo_eurocrypt02}
Y.~Desmedt and Y.~Wang.
\newblock Perfectly secure message transmission revisited.
\newblock In L.~Knudsen, editor, {\em Advances in Cryptology -- EUROCRYPT
  2002}, volume 2332 of {\em Lecture Notes in Computer Science}, pages
  502--517. Springer Berlin Heidelberg, 2002.

\bibitem{dolev_82_BG}
D.~Dolev.
\newblock The byzantine generals strike again.
\newblock {\em Journal of Algorithms}, 3(1):1430, March 1982.

\bibitem{Dolev90perfectlysecure}
D.~Dolev, C.~Dwork, O.~Waarts, and M.~Yung.
\newblock Perfectly secure message transmission.
\newblock {\em Journal of the Association for Computing Machinery (JACM)},
  40(1):17--14, 1993.

\bibitem{impossible_proof_lynch}
M.~J. Fischer, N.~A. Lynch, and M.~Merritt.
\newblock Easy impossibility proofs for distributed consensus problems.
\newblock In {\em Proceedings of the fourth annual ACM symposium on Principles
  of distributed computing}, PODC '85, pages 59--70, New York, NY, USA, 1985.
  ACM.

\bibitem{psl_BG_1982}
L.~Lamport, R.~Shostak, and M.~Pease.
\newblock The byzantine generals problem.
\newblock {\em ACM Trans. on Programming Languages and Systems}, 1982.

\bibitem{Sundaram_journal}
H.~LeBlanc, H.~Zhang, X.~Koutsoukos, and S.~Sundaram.
\newblock Resilient asymptotic consensus in robust networks.
\newblock {\em IEEE Journal on Selected Areas in Communications: Special Issue
  on In-Network Computation}, 31:766--781, April 2013.

\bibitem{leblanc_HiCoNs}
H.~LeBlanc, H.~Zhang, S.~Sundaram, and X.~Koutsoukos.
\newblock Consensus of multi-agent networks in the presence of adversaries
  using only local information.
\newblock {\em HiCoNs}, 2012.

\bibitem{impossible_link}
U.~Schmid, B.~Weiss, and I.~Keidar.
\newblock Impossibility results and lower bounds for consensus under link
  failures.
\newblock {\em SIAM J. Comput.}, 38(5):1912--1951, Jan. 2009.

\bibitem{Shankar_SODA08}
B.~Shankar, P.~Gopal, K.~Srinathan, and C.~P. Rangan.
\newblock Unconditionally reliable message transmission in directed networks.
\newblock In {\em Proceedings of the nineteenth annual ACM-SIAM symposium on
  Discrete algorithms}, SODA '08, pages 1048--1055, Philadelphia, PA, USA,
  2008. Society for Industrial and Applied Mathematics.

\bibitem{Tseng_general}
L.~Tseng and N.~H. Vaidya.
\newblock Iterative approximate byzantine consensus under a generalized fault
  model.
\newblock In {\em In International Conference on Distributed Computing and
  Networking (ICDCN)}, January 2013.

\bibitem{vaidya_incomplete}
N.~H. Vaidya.
\newblock Iterative byzantine vector consensus in incomplete graphs.
\newblock In {\em In International Conference on Distributed Computing and
  Networking (ICDCN)}, January 2014.

\bibitem{vaidya_PODC12}
N.~H. Vaidya, L.~Tseng, and G.~Liang.
\newblock Iterative approximate byzantine consensus in arbitrary directed
  graphs.
\newblock In {\em Proceedings of the thirty-first annual ACM symposium on
  Principles of distributed computing}, PODC '12. ACM, 2012.

\bibitem{Graph_theory_west}
D.~B. West.
\newblock {\em Introduction To Graph Theory}.
\newblock Prentice Hall, 2001.

\bibitem{Sundaram_ACC}
H.~Zhang and S.~Sundaram.
\newblock Robustness of complex networks with implications for consensus and
  contagion.
\newblock In {\em Proceedings of CDC 2012, the 51st IEEE Conference on Decision
  and Control}, 2012.

\bibitem{Sundaram}
H.~Zhang and S.~Sundaram.
\newblock Robustness of distributed algorithms to locally bounded adversaries.
\newblock In {\em Proceedings of ACC 2012, the 31st American Control
  Conference}, 2012.

\end{thebibliography}
\end{document}